\documentclass[11pt]{scrartcl}

\usepackage{url}
\usepackage[hidelinks]{hyperref}
\usepackage[utf8]{inputenc}
\usepackage[small]{caption}
\usepackage{graphicx}
\usepackage{amsmath}
\usepackage{booktabs}
\usepackage{amsthm}
\usepackage{amssymb}
\usepackage{pifont}
\usepackage{cleveref,bm,multirow}
\usepackage{color}
\usepackage{xspace}
\usepackage{tikz}
\usepackage{tabularx}
\usepackage{pbox}
\usepackage{framed}
\usepackage[shortlabels]{enumitem}
\usepackage{thm-restate}
\usepackage{nicematrix}
\usepackage{algorithm}
\usepackage{algorithmic}

\usepackage{natbib}
\usepackage{babel}
\usepackage{subcaption}
\usetikzlibrary{arrows.meta}
\usetikzlibrary{positioning}
\renewcommand{\emptyset}{\varnothing}

\usepackage{cleveref}

\renewcommand{\emptyset}{\varnothing}
\newcommand{\yes}{\textcolor{green!50!black}{\ding{51}}}
\newcommand{\no}{\textcolor{red!50!black}{\ding{55}}}

\newtheorem{theorem}{Theorem}[section]

\newtheorem{proposition}[theorem]{Proposition}
\newtheorem{lemma}[theorem]{Lemma}

\theoremstyle{definition}

\allowdisplaybreaks

\title{Reachability of Fair Allocations via Sequential Exchanges}

\author{
Ayumi Igarashi\\University of Tokyo
\and
Naoyuki Kamiyama\\Kyushu University
\and
Warut Suksompong
\qquad
Sheung Man Yuen\\National University of Singapore
}

\date{\vspace{-5ex}}

\begin{document}

\maketitle

\begin{abstract}
In the allocation of indivisible goods, a prominent fairness notion is envy-freeness up to one good (EF1).
We initiate the study of reachability problems in fair division by investigating the problem of whether one EF1 allocation can be reached from another EF1 allocation via a sequence of exchanges such that every intermediate allocation is also EF1.
We show that two EF1 allocations may not be reachable from each other even in the case of two agents, and deciding their reachability is PSPACE-complete in general.
On the other hand, we prove that reachability is guaranteed for two agents with identical or binary utilities as well as for any number of agents with identical binary utilities.
We also examine the complexity of deciding whether there is an EF1 exchange sequence that is optimal in the number of exchanges required.
\end{abstract}

\section{Introduction}

Fair division refers to the study of how to allocate resources fairly among competing agents, with applications ranging from divorce settlement to university course allocation to international dispute resolution \citep{BramsTa96,Moulin03,Thomson16}.
While its formal study has a long and storied history dating back to the work of \citet{Steinhaus48}, it remains a highly active research area at the intersection of mathematics, economics, and computer science.
In particular, researchers have recently drawn connections between fair division and various other fields such as graph theory \citep{BeiIgLu22,BiloCaFl22}, extremal combinatorics \citep{BerendsohnBoKo22,AkramiAlCh23}, two-sided matching \citep{FreemanMiSh21,IgarashiKaSu23}, and differential privacy \citep{ManurangsiSu23}, to name but a few.

In fair division, the goal is typically to find an allocation of the resource that is ``fair'' with respect to the agents' preferences.
When allocating indivisible goods---such as books, clothes, and office supplies---a prominent fairness notion in the literature is \emph{envy-freeness up to one good (EF1)}.
In an EF1 allocation of the goods, an agent is allowed to envy another agent only if there exists a good in the latter agent's bundle whose removal would eliminate this envy.
The ``up to one good'' relaxation is necessitated by the fact that full envy-freeness is sometimes infeasible, as can be seen when two agents compete for a single valuable good.
It is well-known that an EF1 allocation always exists regardless of the agents' valuations for the goods and can moreover be computed in polynomial time \citep{LiptonMaMo04,Budish11}.
The simplicity, guaranteed existence, and efficient computation makes EF1 a particularly attractive fairness notion.\footnote{By contrast, it remains unknown whether a stronger fairness notion called \emph{envy-freeness up to any good (EFX)} can always be satisfied \citep{AkramiAlCh23}, whereas another well-studied fairness notion, \emph{maximin share fairness (MMS)}, does not offer guaranteed existence \citep{KurokawaPrWa18}.}

In this work, we take a different perspective by initiating the study of \emph{reachability} in fair division.
Given two fair allocations---an initial allocation and a target allocation---we are interested in whether the target allocation can be reached from the initial allocation via a sequence of operations such that every intermediate allocation is also fair.
As an application of our problem, consider a company that wants to redistribute some of its employees between its departments.
Since performing the entire redistribution at once may excessively disrupt the operation of the departments, the company prefers to gradually adjust the distribution while maintaining fairness among the departments throughout the process.
Another example is a museum that plans to reallocate certain exhibits among its branches---performing one small change at a time can help ensure a seamless transition for the visitors.
In this paper, we shall use EF1 as our fairness benchmark and allow any two agents to \emph{exchange} a pair of goods in an operation.
The reachability between EF1 allocations, or lack thereof, is an interesting structural property in itself; similar properties have been studied in other collective decision-making scenarios such as voting \citep{ObraztsovaElFa13,ObraztsovaElFa20}.

Closest to our work is perhaps a line of work initiated by \citet{GourvesLeWi17}.
These authors considered the ``housing market'' setting, where the number of agents is the same as the number of goods and each agent receives exactly one good.
In their model, a pair of agents is allowed to exchange goods if the two agents are neighbors in a given social network and the exchange benefits both agents.
Their paper, along with a series of follow-up papers \citep{HuangXi20,LiPlSi21,MullerBe21,ItoKaKa23}, explored the complexity of determining whether an allocation can be reached from another allocation in this model and its variants. 
More broadly, reachability problems are also known as \emph{reconfiguration} problems \citep{Nishimura18}; examples of such problems that have been studied include minimum spanning tree \citep{ItoDeHa11}, graph coloring \citep{JohnsonKrKr16}, and perfect matching \citep{BonamyBoHe19}.

\subsection{Our Contributions}

\begin{table*}[t]
\centering
\begin{NiceTabular}{c|c||c|c|c|c}
\multicolumn{2}{c||}{utilities} & general & identical & binary & identical + binary \\ \hline \hline
\multirow{2}{*}{$n=2$} 
& connected? & \no \ (Th.~\ref{thm:gen_2_connected}) 
             & \yes \ (Th.~\ref{thm:iden_2_optimal}) 
             & \yes \ (Th.~\ref{thm:binary_2_optimal}) 
             & \yes \ (Th.~\ref{thm:iden_2_optimal}/\ref{thm:binary_2_optimal}) \\
& optimal?   & \no \ (Th.~\ref{thm:gen_2_optimal}) 
             & \yes \ (Th.~\ref{thm:iden_2_optimal}) 
             & \yes \ (Th.~\ref{thm:binary_2_optimal}) 
             & \yes \ (Th.~\ref{thm:iden_2_optimal}/\ref{thm:binary_2_optimal}) \\ \hline
\multirow{2}{*}{$n \geq 3$} 
& connected? & \no \ (Th.~\ref{thm:binary_3_connected}/\ref{thm:iden_3_connected}) 
             & \no \ (Th.~\ref{thm:iden_3_connected}) 
             & \no \ (Th.~\ref{thm:binary_3_connected}) 
             & \yes \ (Th.~\ref{thm:idenbin_2_connected}) \\
& optimal?   & \no \ (Th.~\ref{thm:idenbin_3_optimal}) 
             & \no \ (Th.~\ref{thm:idenbin_3_optimal}) 
             & \no \ (Th.~\ref{thm:idenbin_3_optimal}) 
             & \no \ (Th.~\ref{thm:idenbin_3_optimal})
\end{NiceTabular}
\caption{Overview of our results. 
The top row indicates the class of utility functions considered, ``connected?''~refers to whether the EF1 exchange graph is always connected, and ``optimal?''~refers to whether there always exists an optimal EF1 exchange path between any two EF1 allocations provided that the EF1 exchange graph is connected.} \label{tab:results}
\end{table*}

As is often done in fair division, we assume that every agent is equipped with an additive utility function.
We consider an ``exchange graph'' with allocations as vertices.
The first question we study is whether it is always possible for agents to reach a target EF1 allocation from an initial EF1 allocation by exchanging goods sequentially with each other while maintaining the EF1 property in all the intermediate allocations; in other words, we ask whether the subgraph of the exchange graph consisting of all EF1 allocations is \emph{connected}. 
The second question is whether we could perform this exchange process using as few exchanges as would be required if the intermediate allocations need not be EF1; that is, whether there exists an EF1 exchange path which is \emph{optimal} in terms of the number of exchanges required.
Note that each agent's bundle size remains unchanged throughout the process since every operation is an exchange of goods. 
Our formal model is described in \Cref{sec:prelim}.

In \Cref{sec:two_agents}, we investigate the setting where there are only two agents. 
Perhaps surprisingly, we establish negative results even for this setting: the EF1 exchange graph may not be connected, and even for those instances in which it is connected, optimal EF1 exchange paths may not exist between EF1 allocations. 
Therefore, we consider restricted classes of utility functions. 
We show that an optimal EF1 exchange path always exists between any two EF1 allocations if the utilities are identical \emph{or} binary; this implies the connectivity of the EF1 exchange graph in these cases as well.

In \Cref{sec:three_agents}, we explore the general setting of three or more agents. 
Interestingly, we show that finding the smallest number of exchanges between two allocations is NP-hard in this setting even if we disregard the EF1 restriction. 
In addition, we establish that deciding whether an EF1 exchange path exists between two allocations is PSPACE-complete, and deciding whether an optimal such path exists is NP-hard even for four agents with identical utilities. 
We also examine restricted utility functions in more detail.
We show that while connectivity of the EF1 exchange graph is guaranteed for identical binary utilities, the same holds neither for identical utilities nor for binary utilities separately. Furthermore, the optimality of EF1 exchange paths cannot be guaranteed even for identical binary utilities.
Overall, our findings demonstrate that the case of three or more agents is much less tractable than that of two agents in our setting.

With the exception of hardness results (Theorems \ref{thm:dist_nphard}, \ref{thm:gen_2_connected_pspace}, and \ref{thm:iden_4_optimal_nphard}), our results are summarized in \Cref{tab:results}.
For the positive results, we also show that the corresponding exchange paths can be found in polynomial time.
Additionally, in \Cref{app:transfer}, we present results for an alternative setting where \emph{transferring} goods is allowed instead of, or in addition to, exchanging them.

\section{Preliminaries} \label{sec:prelim}

Let $N$ be a set of $n \geq 2$ agents, and $M$ be a set of $m \geq 1$ goods.
We typically denote the agents by $1, \dots ,n$ and the goods by $g_1, \dots ,g_m$.
A \emph{bundle} is a (possibly empty) subset of goods.
An \emph{allocation} $\mathcal{A} = (A_1, \ldots, A_n)$ is an ordered partition of $M$ into $n$ bundles such that bundle $A_i$ is allocated to agent $i \in N$.
An \emph{(allocation) size vector} $\vec{s} = (s_1, \ldots, s_n)$ is a vector of non-negative integers such that $\sum_{i \in N} s_i = m$ and $s_i = |A_i|$ for all $i \in N$.

Given $N$, $M$, and $\vec{s}$, define the \emph{exchange graph} $G = G(N, M, \vec{s})$ as a simple undirected graph with the following properties: the set of vertices consists of all allocations~$\mathcal{A}$ with size vector $\vec{s}$, and the set of edges consists of all pairs $\{\mathcal{A}, \mathcal{B}\}$ such that $\mathcal{B} = (B_1, \ldots, B_n)$ can be obtained from $\mathcal{A} = (A_1, \ldots, A_n)$ by having two agents exchange one pair of goods with each other---that is, there exist distinct agents $i, j \in N$ and goods $g \in A_i$ and $g' \in A_j$ such that $B_i = (A_i \cup \{g'\}) \setminus \{g\}$, $B_j = (A_j \cup \{g\}) \setminus \{g'\}$, and $B_k = A_k$ for all $k \in N \setminus \{i, j\}$. 
Note that the exchange graph is a non-empty connected graph. 
A path from one allocation to another on the graph is called an \emph{exchange path}. 
The \emph{distance} between two allocations is the length of a shortest exchange path between them.

Each agent $i \in N$ has a \emph{utility function} $u_i : 2^M \to \mathbb{R}_{\geq 0}$ that maps bundles to non-negative real numbers.
We write $u_i (g)$ instead of $u_i (\{g\})$ for a single good $g \in M$, and assume that the utility functions are additive, i.e., $u_i (M') = \sum_{g \in M'} u_i (g)$ for all $i \in N$ and $M' \subseteq M$. 
The utility functions are \emph{identical} if $u_i = u_j$ for all $i, j \in N$---we shall use $u$ to denote the common utility function in this case.
The utility functions are \emph{binary} if $u_i(g) \in \{0, 1\}$ for all $i \in N$ and $g \in M$.
An allocation $\mathcal{A}$ is \emph{envy-free up to one good (EF1)} if for all pairs $i, j \in N$ such that $A_j \neq \emptyset$, there exists a good $g \in A_j$ such that $u_i(A_i) \geq u_i(A_j \setminus \{g\})$.

Given $N$, $M$, $\vec{s}$, and $(u_i)_{i \in N}$, define the \emph{EF1 exchange graph} $H = H(N, M, \vec{s}, (u_i)_{i \in N})$ as the subgraph of the exchange graph $G$ induced by all EF1 allocations, i.e., $H$ contains all vertices in $G$ that correspond to EF1 allocations and all edges in $G$ incident to two EF1 allocations. 
As we shall see later, EF1 exchange graphs are not always connected, unlike exchange graphs.
An exchange path using only the edges in~$H$ is called an \emph{EF1 exchange path}. 
An EF1 exchange path is \emph{optimal} if its length is equal to the distance between the two corresponding allocations (in $G$). 

An \emph{instance} consists of a set of agents $N$, a set of goods $M$, a size vector $\vec{s}$, and agents' utility functions $(u_i)_{i \in N}$.

\section{Two Agents} \label{sec:two_agents}

In this section, we examine properties of the EF1 exchange graph when there are only two agents.
We remark that this is an important special case in fair division and has been the focus of several prior papers in the area.\footnote{\citet[Sec.~1.1.1]{PlautRo20} discussed the significance of the two-agent setting in detail.}

We first consider the question of whether the EF1 exchange graph is necessarily connected. One may intuitively think that with only two agents, an EF1 exchange path is guaranteed between any two EF1 allocations because the two agents only need to consider the envy between themselves.
An agent may then carefully select a good from her bundle to exchange with the other agent so as to ensure that the subsequent allocation is also EF1. 
However, this in fact cannot always be done, as our first result shows.

\begin{theorem} \label{thm:gen_2_connected}
There exists an instance with $n = 2$ agents with the same ordinal preferences over the goods such that the EF1 exchange graph is disconnected.
\end{theorem}

\begin{proof}
Consider the utility of the goods as follows:
\begin{center}
\begin{tabular}{c|cccccccc}
$g$      & $g_1$ & $g_2$ & $g_3$ & $g_4$ & $g_5$ & $g_6$ & $g_7$ & $g_8$ \\ \hline
$u_1(g)$ & $3$   & $3$   & $2$   & $2$   & $2$   & $2$   & $0$   & $0$   \\
$u_2(g)$ & $3$   & $3$   & $1$   & $1$   & $1$   & $1$   & $0$   & $0$   \\
\end{tabular}
\end{center}
Let $\mathcal{A}$ and $\mathcal{B}$ be allocations such that $A_1 = B_2 = \{g_1, g_2, g_7, g_8\}$ and $A_2 = B_1 = \{g_3, g_4, g_5, g_6\}$---it can be verified that both $\mathcal{A}$ and $\mathcal{B}$ are EF1. 
If there exists an EF1 exchange path between $\mathcal{A}$ and $\mathcal{B}$, then there exists an EF1 allocation $\mathcal{A}'$ adjacent to $\mathcal{A}$ on the exchange path. 
Without loss of generality, $\mathcal{A}'$ can be reached from $\mathcal{A}$ by exchanging $g_3$ with either $g_1$ or $g_7$. 
If $g_3$ is exchanged with $g_1$, then agent~$1$ envies agent~$2$ by more than one good. 
If $g_3$ is exchanged with $g_7$, then agent~$2$ envies agent~$1$ by more than one good.
Therefore, neither of these exchanges leads to an EF1 allocation, so $\mathcal{A}'$ cannot be EF1. Hence, no EF1 exchange path exists between $\mathcal{A}$ and $\mathcal{B}$.
\end{proof}

Next, we consider the question of whether an \emph{optimal} EF1 exchange path always exists between two EF1 allocations. 
By \Cref{thm:gen_2_connected}, even an EF1 exchange path may not exist, so an optimal such path does not necessarily exist either.
We therefore focus on instances in which the EF1 exchange graph is \emph{connected}. 
It turns out that even for such instances, an optimal EF1 exchange path still might not exist.

\begin{theorem}  \label{thm:gen_2_optimal}
There exists an instance with $n = 2$ agents satisfying the following properties: the EF1 exchange graph is connected, but for some pair of EF1 allocations, no optimal EF1 exchange path exists between them.
\end{theorem}

\begin{proof}
Consider $\vec{s} = (3, 3)$ and the utility of the goods as follows:
\begin{center}
\begin{tabular}{c|cccccc}
$g$      & $g_1$ & $g_2$ & $g_3$ & $g_4$ & $g_5$ & $g_6$ \\ \hline
$u_1(g)$ & $5$   & $3$   & $1$   & $0$   & $2$   & $2$   \\
$u_2(g)$ & $0$   & $3$   & $1$   & $5$   & $2$   & $2$   \\ 
\end{tabular}
\end{center}
Let $\mathcal{B}$ be the allocation such that $B_1 = \{g_1, g_2, g_3\}$ and $B_2 = \{g_4, g_5, g_6\}$---it can be verified that $\mathcal{B}$ is EF1.
We first prove that the EF1 exchange graph is connected by constructing an EF1 exchange path between any EF1 allocation $\mathcal{A}$ and the EF1 allocation~$\mathcal{B}$.
If $g_1$ is not with agent~$1$ or $g_4$ is not with agent~$2$ in~$\mathcal{A}$, perform any exchange involving $g_1$ and/or $g_4$ so that $g_1$ is now with agent $1$ and $g_4$ is now with agent~$2$. 
After the exchange, for each $i \in \{1, 2\}$, agent $i$'s bundle is worth at least $5$ to her, while any two goods in agent $(3 - i)$'s bundle are worth at most $5$ to agent $i$, so the allocation is EF1. 
Now, we can exchange the goods in $\{g_2, g_3, g_5, g_6\}$ in an arbitrary order to reach $\mathcal{B}$ after at most two more exchanges.

We next prove that an optimal EF1 exchange path between allocations $\mathcal{C}$ and $\mathcal{D}$ does not exist, where $C_1 = \{g_2, g_3, g_4\}$, $C_2 = \{g_1, g_5, g_6\}$, $D_1 = \{g_4, g_5, g_6\}$, and $D_2 = \{g_1, g_2, g_3\}$; it can be verified that both $\mathcal{C}$ and $\mathcal{D}$ are EF1, and the distance between $\mathcal{C}$ and $\mathcal{D}$ is $2$ (through exchanging $g_2 \leftrightarrow g_5$ and $g_3 \leftrightarrow g_6$). 
Suppose there exists an optimal EF1 exchange path between $\mathcal{C}$ and $\mathcal{D}$, and let $\mathcal{C}'$ be the EF1 allocation between $\mathcal{C}$ and~$\mathcal{D}$ on the exchange path. 
Since $\mathcal{C}$ and $\mathcal{C}'$ are adjacent, one good from $\{g_2, g_3\}$ must be exchanged with one good from $\{g_5, g_6\}$ in $\mathcal{C}$ to reach $\mathcal{C}'$. 
However, no matter which goods are exchanged with this restriction, there exists $i \in \{1, 2\}$ such that agent $i$'s bundle is worth $3$ to her and agent $(3 - i)$'s bundle is worth $5 + 5$ to agent $i$, contradicting the EF1 property of $\mathcal{C}'$. 
Therefore, no optimal EF1 exchange path exists between $\mathcal{C}$ and~$\mathcal{D}$. 
\end{proof}

In light of these negative results, we turn our attention to special classes of utility functions: identical utilities and binary utilities. 
We prove that for these two classes of utility functions, the EF1 exchange graph is always connected, and moreover, an optimal EF1 exchange path exists between every pair of EF1 allocations.

\begin{theorem} \label{thm:iden_2_optimal}
Let an instance with $n = 2$ agents and identical utilities be given. 
Then, the EF1 exchange graph is connected. 
Moreover, there exists an optimal EF1 exchange path between any two EF1 allocations, and this path can be computed in polynomial time.
\end{theorem}

\begin{theorem} \label{thm:binary_2_optimal}
Let an instance with $n = 2$ agents and binary utilities be given. 
Then, the EF1 exchange graph is connected. 
Moreover, there exists an optimal EF1 exchange path between any two EF1 allocations, and this path can be computed in polynomial time.
\end{theorem}

To establish these results, we shall prove by induction on $t$ that two EF1 allocations with distance $t$ have an optimal EF1 exchange path between them. 
For the base case $t = 0$, an optimal EF1 exchange path trivially exists. 
For the inductive step, let $t \geq 1$ be given, and assume the inductive hypothesis that any two EF1 allocations with distance $t - 1$ have an EF1 exchange path of length $t - 1$. 
Now, let $\mathcal{A} = (A_1, A_2)$ and $\mathcal{B} = (B_1, B_2)$ be any two EF1 allocations with distance $t$; this means that $|A_1 \setminus B_1| = |A_2 \setminus B_2| = t$.
Define $X = A_1 \setminus B_1 = \{x_1, \ldots, x_t\}$ and $Y = A_2 \setminus B_2 = \{y_1, \ldots, y_t\}$. 
We show that there exist goods $x_k \in X$ and $y_\ell \in Y$ such that exchanging them in $\mathcal{A}$ leads to an EF1 allocation $\mathcal{A}' = (A'_1, A'_2)$. 
If this is possible, then $|A'_1 \setminus B_1| = |A'_2 \setminus B_2| = t - 1$, which implies that the distance between $\mathcal{A}'$ and $\mathcal{B}$ is $t - 1$. 
By the inductive hypothesis, there exists an EF1 exchange path between $\mathcal{A}'$ and $\mathcal{B}$ of length $t - 1$. 
This means that there exists an EF1 exchange path between $\mathcal{A}$ and $\mathcal{B}$ via $\mathcal{A}'$ of length $t$, which is optimal, hence completing the proof.

For the time complexity, for each pair of goods from \mbox{$X \times Y$}, one can check in polynomial time whether exchanging them leads to an EF1 allocation. 
Since there are at most $t^2$ pairs of goods to check at each step, and there are $t$ steps in the path, the running time claim follows.

\begin{proof}[Proof of \Cref{thm:iden_2_optimal}]
We follow the notation and inductive outline described before this proof.
Assume that the goods in $X$ and $Y$ are arranged in non-increasing order of utilities, i.e., $u(x_i) \geq u(x_j)$ and $u(y_i) \geq u(y_j)$ whenever $i < j$. 
Denote $\Delta_k = u(y_k) - u(x_k)$ for all $k \in \{1, \ldots, t\}$. 
Define $A_1' = (A_1 \cup \{y_1\}) \setminus \{x_1\}$ and $A_2' = (A_2 \cup \{x_1\}) \setminus \{y_1\}$ to be the bundles after exchanging $x_1$ and $y_1$. 
If $(A_1', A_2')$ is EF1, we are done by induction.
Otherwise, we assume without loss of generality that in the allocation $(A_1', A_2')$, agent~$2$ envies agent~$1$ by more than one good. 
Let $x$ be a highest-utility good in $A_1$---we may assume that $x \neq x_k$ for all $k \geq 2$. Since $(A_1, A_2)$ is an EF1 allocation, we have $u(x) \geq \gamma := u(A_1) - u(A_2)$.

If both $x = x_1$ and $\Delta_1 < 0$ are true, then 
\begin{align*}
    u(A_2')
    = u(A_2) - \Delta_1 
    > u(A_2) 
    \geq u(A_1 \setminus \{x\}) 
    = u(A_1 \setminus \{x_1\}) 
    &= u(A_1' \setminus \{y_1\}),
\end{align*}
which shows that agent $2$ does not envy agent $1$ by more than one good in $(A_1', A_2')$---a contradiction. 
Therefore, we must have $x \neq x_1$ or $\Delta_1 \geq 0$.
If $x \neq x_1$, then both $x$ and $y_1$ belong to $A_1'$.
If $x = x_1$ and $\Delta_1 \ge 0$, then $y_1$ belongs to $A_1'$ and $u(y_1) \ge u(x)$.
Hence, in either case, we have
\begin{align*}
    \max \{u(x), u(y_1)\} < u(A_1') - u(A_2') 
    &= u(A_1) - u(A_2) + 2\Delta_1,
\end{align*}
which implies
\begin{align} \label{eq:k_bound}
    \gamma + 2\Delta_1 > \max \{u(x), u(y_1)\}.
\end{align}

We claim that there exists $k \in \{2, \ldots, t\}$ such that 
\begin{align} \label{eq:ineq_1}
    2\Delta_k \leq u(x) - \gamma .
\end{align}
Suppose on the contrary that $2\Delta_k > u(x) - \gamma$ for all $k \in \{2, \ldots, t\}$.
Since every good in $A_1$ has value at most $u(x)$ and every good in $B_1\setminus A_1$ has value at most $u(y_1)$, it holds that every good in $B_1$ has value at most $\max\{u(x), u(y_1)\}$.
As $(B_1, B_2)$ is an EF1 allocation, we have
\begin{align*}
    \max \{u(x), u(y_1)\} \geq u(B_1) - u(B_2) 
    &= (u(A_1) - u(A_2)) + \sum_{k=1}^t 2\Delta_k \\
    &= \gamma + 2\Delta_1 + \sum_{k=2}^t 2\Delta_k \\
    &\geq \gamma + 2\Delta_1 + \sum_{k=2}^t (u(x) - \gamma) 
    \geq \gamma + 2\Delta_1,
\end{align*}
where the last inequality holds because $u(x) \geq \gamma$ and \mbox{$t\ge 1$}.
This contradicts \eqref{eq:k_bound}. 
Therefore, let $k \in \{2, \ldots, t\}$ be an index that satisfies \eqref{eq:ineq_1}. We now claim that we must have
\begin{align} \label{eq:ineq_2}
    2\Delta_k \geq \max \{u(x), u(y_1)\} - 2u(y_1) - \gamma.
\end{align} 
Suppose on the contrary that $2\Delta_k < \max \{u(x), u(y_1)\} - 2u(y_1) - \gamma$. Then we have
\begin{align*}
    \max \{u(x), u(y_1)\} - 2u(y_1) - \gamma 
    &> 2\Delta_k \tag*{(by assumption)} \\
    &\geq -2u(x_k) \tag*{(since $u(y_k) \geq 0$)} \\
    &\geq -2u(x_1), \tag*{(since $u(x_k) \leq u(x_1)$)}
\end{align*}
which implies
\begin{align*}
    \gamma + 2\Delta_1 < \max \{u(x), u(y_1)\},
\end{align*}
contradicting \eqref{eq:k_bound}.
This establishes \eqref{eq:ineq_2}.

Combining inequalities \eqref{eq:ineq_1} and \eqref{eq:ineq_2}, we have
\begin{align*} 
    -u(y_1) &\leq \max \{u(x) - u(y_1), 0\} - u(y_1) \\ 
    &= \max \{u(x), u(y_1)\} - 2u(y_1) \\
    &\leq \gamma + 2\Delta_k \tag*{(by \eqref{eq:ineq_2})} \\ 
    &\leq u(x), \tag*{(by \eqref{eq:ineq_1})}
\end{align*}
which implies $\gamma + 2\Delta_k \in [-u(y_1), u(x)]$.
We claim that exchanging $x_k$ and $y_k$ results in an EF1 allocation, i.e., the allocation comprising $A_1'' = (A_1 \cup \{y_k\}) \setminus \{x_k\}$ and $A_2'' = (A_2 \cup \{x_k\}) \setminus \{y_k\}$ is EF1. This is because
\begin{align*}
    u(A_1'') - u(A_2'') = u(A_1) - u(A_2) + 2\Delta_k 
    &= \gamma + 2\Delta_k 
    \in [-u(y_1), u(x)],
\end{align*}
where $x \in A_1''$ and $y_1 \in A_2''$---note that $x$ ($\neq x_k$) and $y_1$ were not exchanged going from $\mathcal{A}$ to $\mathcal{A}''$. 
This completes the induction and therefore the proof.
\end{proof}

\begin{proof}[Proof of \Cref{thm:binary_2_optimal}]
We follow the notation and inductive outline described before the proof of \Cref{thm:iden_2_optimal}.
Recall that $X = A_1 \setminus B_1 = \{x_1, \ldots, x_t\}$ and $Y = A_2 \setminus B_2 = \{y_1, \ldots, y_t\}$.
Let $M_i = \{g \in M \mid u_i(g) = 1\}$ for $i \in \{1, 2\}$. 
Note that if $|A_i \cap M_i| > |A_{3-i} \cap M_i|$, then agent~$i$ does not envy agent $(3 - i)$ by more than one good after the exchange of any pair of goods. 
Therefore, if $|A_i \cap M_i| > |A_{3-i} \cap M_i|$ is true for both $i \in \{1, 2\}$, then exchanging any pair of goods from $X$ and $Y$ works.

Otherwise, suppose that $|A_i \cap M_i| \leq |A_{3-i} \cap M_i|$ is true for some $i \in \{1, 2\}$, and without loss of generality, assume that $i = 1$. 
We claim that $|X \cap M_1| \leq |Y \cap M_1|$.
Suppose by way of contradiction that $|X \cap M_1| > |Y \cap M_1|$. Then, 
\begin{align*}
    |B_1 \cap M_1| &= |((A_1 \setminus X) \cup Y) \cap M_1| \\
    &= |A_1 \cap M_1| - |X \cap M_1| + |Y \cap M_1| \\
    &\leq |A_2 \cap M_1| - 1 \\
    &= |((B_2 \setminus X) \cup Y) \cap M_1| - 1 \\
    &= |B_2 \cap M_1| - |X \cap M_1| + |Y \cap M_1| - 1 
    \leq |B_2 \cap M_1| - 2,
\end{align*}
which means that agent $1$ envies agent $2$ by more than one good in $\mathcal{B}$, contradicting the assumption that $\mathcal{B}$ is an EF1 allocation. 
Therefore, we must have $|X \cap M_1| \leq |Y \cap M_1|$. 
Thus, there exists a bijection $\phi : X \to Y$ such that $u_1(x) \leq u_1(\phi(x))$ for all $x \in X$; this can be obtained by ensuring that goods in $X \cap M_1$ are mapped to goods in $Y \cap M_1$.
Exchanging $x$ and $\phi(x)$ in $\mathcal{A}$ for any $x \in X$ will not make agent $1$ envy agent $2$ by more than one good.

Now, we consider two cases for agent $2$. 
If $|A_2 \cap M_2| > |A_1 \cap M_2|$, then agent $2$ does not envy agent~$1$ by more than one good after the exchange of any pair of goods. 
In particular, we can exchange $x$ and $\phi(x)$ for any $x \in X$, and we are done by induction. 
In the other case, we have $|A_2 \cap M_2| \leq |A_1 \cap M_2|$.
We claim that there exists some $x \in X$ such that $u_2(\phi(x)) \leq u_2(x)$.
Suppose on the contrary that $u_2(\phi(x)) > u_2(x)$ for all $x \in X$. 
This means that $u_2(x) = 0$ for all $x \in X$ and $u_2(y) = 1$ for all $y \in Y$, and hence $|Y \cap M_2| - |X \cap M_2| = t \geq 1$, where $t = |X| = |Y|$. 
Then, we have
\begin{align*}
    |B_2 \cap M_2| &= |((A_2 \setminus Y) \cup X) \cap M_2| \\
    &= |A_2 \cap M_2| - |Y \cap M_2| + |X \cap M_2| \\
    &\leq |A_1 \cap M_2| - 1 \\
    &= |((B_1 \setminus Y) \cup X) \cap M_2| - 1 \\
    &= |B_1 \cap M_2| - |Y \cap M_2| + |X \cap M_2| - 1 
    \leq |B_1 \cap M_2| - 2,
\end{align*}
which means that agent $2$ envies agent $1$ by more than one good in $\mathcal{B}$, contradicting the assumption that $\mathcal{B}$ is an EF1 allocation. 
Thus, $u_2(\phi(x_k)) \leq u_2(x_k)$ for some $x_k \in X$. 
In particular, exchanging $x_k$ and $\phi(x_k)$ in $\mathcal{A}$ does not make agent $2$ envy agent $1$ by more than one good.
Hence, exchanging $x_k \in X$ and $\phi(x_k) \in Y$ leads to an EF1 allocation, completing the induction and therefore the proof.
\end{proof}

Since the EF1 exchange graph $H$ is a subgraph of the exchange graph $G$, the distance between two allocations (in $G$) cannot be greater than the length of the shortest EF1 exchange path between the two allocations in $H$.
In Theorems \ref{thm:iden_2_optimal} and \ref{thm:binary_2_optimal}, the polynomial-time algorithms find EF1 exchange paths in $H$ that are optimal in the exchange graph~$G$; such exchange paths must also be the shortest possible ones in $H$.

\section{Three or More Agents} \label{sec:three_agents}

In this section, we address the general case where there are more than two agents. 
We shall see that this case is less tractable, both existentially and computationally.

Before we present our results on the EF1 exchange graph, we provide some insights on finding the distance between two allocations regardless of EF1 considerations.  
Observe that finding this distance for two agents is simple, as the distance equals the number of goods from each of the two bundles that need to be exchanged.
However, this task is not so trivial for more agents---in fact, we shall show that it is NP-hard.
To this end, we draw an interesting connection between this distance and the maximum number of disjoint cycles in a graph constructed based on the allocations. 
We start off by detailing how to construct such a graph.

Let $N$, $M$, and $\vec{s}$ be given, and let $\mathcal{A} = (A_1, \ldots, A_n)$ and $\mathcal{B} = (B_1, \ldots, B_n)$ be two allocations with size vector $\vec{s}$. 
Define $G_{\mathcal{A}, \mathcal{B}} = (V_{\mathcal{A}, \mathcal{B}}, E_{\mathcal{A}, \mathcal{B}})$ as a directed multigraph consisting of a set of vertices $V_{\mathcal{A}, \mathcal{B}} = N$ and a set of (directed) edges $E_{\mathcal{A}, \mathcal{B}} = \{e_1, \ldots, e_m\}$. 
For each $k \in \{1, \ldots, m\}$, the edge $e_k$ represents the good $g_k$, and $e_k = (i, j)$ if and only if $g_k \in A_i \cap B_j$, i.e., $g_k$ is in agent $i$'s bundle in $\mathcal{A}$ and in agent $j$'s bundle in $\mathcal{B}$ (possibly $i = j$).
Note that for every vertex~$i$, its indegree is equal to its outdegree, which is equal to $s_i$, the number of goods in agent $i$'s bundle.
Let $\mathfrak{C}_{\mathcal{A}, \mathcal{B}}$ be the collection of partitions of $E_{\mathcal{A}, \mathcal{B}}$ into directed circuits.\footnote{Recall that a directed circuit is a non-empty walk such that the first vertex and the last vertex coincide; we consider a self-loop to be a directed circuit as well.} 
Note that $\mathfrak{C}_{\mathcal{A}, \mathcal{B}}$ is non-empty---for example, a partition of $E_{\mathcal{A}, \mathcal{B}}$ into directed circuits can be constructed in the following way: start with any vertex with an outdegree of at least~$1$, traverse a path until some vertex $v$ is encountered for the second time, remove the resulting directed cycle from $v$ to itself, and repeat on the remaining graph; the remaining graph still satisfies the condition that every vertex has its indegree equal to its outdegree.
Let $c^*_{\mathcal{A}, \mathcal{B}} = \max_{C_{\mathcal{A}, \mathcal{B}} \in \mathfrak{C}_{\mathcal{A}, \mathcal{B}}} |C_{\mathcal{A}, \mathcal{B}}|$ be the maximum cardinality of such a partition. 
Note that a partition with the maximum cardinality must consist only of directed \emph{cycles}; otherwise, if it contains a circuit that passes through a vertex more than once, we can break the circuit into two smaller circuits, contradicting the fact that this partition has the maximum cardinality. 
We claim that the distance between allocations $\mathcal{A}$ and $\mathcal{B}$ is $m - c^*_{\mathcal{A}, \mathcal{B}}$.

\begin{proposition} \label{lem:distance_between_allocations}
Let $N$, $M$, and $\vec{s}$ be given, and let $\mathcal{A}$ and $\mathcal{B}$ be two allocations with size vector $\vec{s}$. 
Then, the distance between $\mathcal{A}$ and $\mathcal{B}$ is $m - c^*_{\mathcal{A}, \mathcal{B}}$.
\end{proposition}

\begin{proof}
We have $m \geq c^*_{\mathcal{A}, \mathcal{B}}$ since every partition in $\mathfrak{C}_{\mathcal{A}, \mathcal{B}}$ is a partition of a set with cardinality $m$, so $m - c^*_{\mathcal{A}, \mathcal{B}} \geq 0$. 
We shall prove the result by strong induction on $m - c^*_{\mathcal{A}, \mathcal{B}}$. 
For the base case, let $\mathcal{A}$ and $\mathcal{B}$ be given such that $m - c^*_{\mathcal{A}, \mathcal{B}} = 0$. 
This means that there exists a partition $C_{\mathcal{A}, \mathcal{B}} \in \mathfrak{C}_{\mathcal{A}, \mathcal{B}}$ such that $|C_{\mathcal{A}, \mathcal{B}}| = m = |E_{\mathcal{A}, \mathcal{B}}|$. 
The only way that this is possible is when every edge in $E_{\mathcal{A}, \mathcal{B}}$ is a self-loop. 
Thus, each good appears in the same agent's bundle in $\mathcal{A}$ and $\mathcal{B}$. 
This means that $\mathcal{A} = \mathcal{B}$, and so the distance between $\mathcal{A}$ and $\mathcal{B}$ is $0 = m - c^*_{\mathcal{A}, \mathcal{B}}$.

For the inductive hypothesis, suppose that there exists a non-negative integer $p_0$ such that for all pairs of allocations $\mathcal{A}$ and $\mathcal{B}$ satisfying $m - c^*_{\mathcal{A}, \mathcal{B}} = p$ for any $p \in \{0, \ldots, p_0\}$, the distance between $\mathcal{A}$ and $\mathcal{B}$ is $p$. 
For the inductive step, consider a pair of allocations $\mathcal{A}$ and $\mathcal{B}$ such that $m - c^*_{\mathcal{A}, \mathcal{B}} = p_0 + 1$. 
We shall prove that the distance between $\mathcal{A}$ and~$\mathcal{B}$ is $p_0 + 1$. 

We first prove that the distance between $\mathcal{A}$ and $\mathcal{B}$ is at most $p_0 + 1$. 
Let $C_{\mathcal{A}, \mathcal{B}} \in \mathfrak{C}_{\mathcal{A}, \mathcal{B}}$ be such that $m - |C_{\mathcal{A}, \mathcal{B}}| = p_0 + 1 > 0$. Since $|E_{\mathcal{A}, \mathcal{B}}| > |C_{\mathcal{A}, \mathcal{B}}|$, $C_{\mathcal{A}, \mathcal{B}}$ contains at least one directed circuit of length at least two. 
For notational simplicity, let one such directed circuit be $v_0 \to v_1 \to \cdots \to v_\ell = v_0$ for some $\ell \geq 2$, where $e_k = (v_{k-1}, v_k)$ for $k \in \{1, \ldots, \ell\}$ without loss of generality. 
From $\mathcal{A}$, exchange $g_{\ell-1}$ (in agent $v_{\ell-2}$'s bundle) with $g_\ell$ (in agent $v_{\ell-1}$'s bundle) to form the allocation $\mathcal{A}'$. 
This removes the directed circuit $v_0 \to v_1 \to \cdots \to v_\ell$ but introduces two new directed circuits: $v_{\ell-1} \to v_{\ell-1}$ and $v_0 \to v_1 \to \cdots \to v_{\ell-2} \to v_0$. 
Thus, there exists a partition $C_{\mathcal{A}', \mathcal{B}} \in \mathfrak{C}_{\mathcal{A}', \mathcal{B}}$ such that $|C_{\mathcal{A}', \mathcal{B}}| = |C_{\mathcal{A}, \mathcal{B}}| + 1$. 
This gives $m - |C_{\mathcal{A}', \mathcal{B}}| = p_0$, and thus $m - c^*_{\mathcal{A}', \mathcal{B}} \leq p_0$. 
By the inductive hypothesis, the distance between $\mathcal{A}'$ and $\mathcal{B}$ is at most $p_0$. 
Therefore, the distance between $\mathcal{A}$ and $\mathcal{B}$ (via $\mathcal{A}'$) is at most $p_0 + 1$, which is $m - c^*_{\mathcal{A}, \mathcal{B}}$.

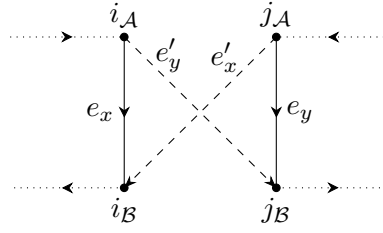
\begin{figure}[htb]
\centering
  \begin{subfigure}[b]{.47\textwidth}
    \centering 
    \begin{tikzpicture}
    \usetikzlibrary{positioning}
    \newdimen\nodeDist
    \nodeDist=10mm
    \newcommand{\arrow}{ \tikz \draw[-{Stealth[length=1.5mm, width=1.5mm]}] (-1pt,0) -- (1pt,0); }
    \node at (0,0) (v0) {};
    \node at (-\nodeDist,\nodeDist) (via) {};
    \node at (-\nodeDist,-\nodeDist) (vib) {};
    \node at (\nodeDist,\nodeDist) (vja) {};
    \node at (\nodeDist,-\nodeDist) (vjb) {};
    \node at (-2.5*\nodeDist,\nodeDist) (via2) {};
    \node at (-2.5*\nodeDist,-\nodeDist) (vib2) {};
    \node at (2.5*\nodeDist,\nodeDist) (vja2) {};
    \node at (2.5*\nodeDist,-\nodeDist) (vjb2) {};
    \draw [black, fill=black] (via) circle [radius=0.05] node[above] {$i_\mathcal{A}$};
    \draw [black, fill=black] (vib) circle [radius=0.05] node[below] {$i_\mathcal{B}$};
    \draw [black, fill=black] (vja) circle [radius=0.05] node[above] {$j_\mathcal{A}$};
    \draw [black, fill=black] (vjb) circle [radius=0.05] node[below] {$j_\mathcal{B}$};
    \draw [] (via.center) -- node[left] {$e_x$} (vib.center) node[sloped, midway, allow upside down]{\arrow};;
    \draw [] (vja.center) -- node[right] {$e_y$} (vjb.center) node[sloped, midway, allow upside down]{\arrow};;
    \draw [-Stealth, dashed] (via.center) -- node[above right=4mm and 0mm] {$e_x'$} (vjb.center);
    \draw [-Stealth, dashed] (vja.center) -- node[above left=4mm and 1mm] {$e_y'$} (vib.center);
    \draw [dotted] (via2.center) -- (via.center) node[sloped, midway, allow upside down]{\arrow};;
    \draw [dotted] (vib.center) -- (vib2.center) node[sloped, midway, allow upside down]{\arrow};;
    \draw [dotted] (vja2.center) -- (vja.center) node[sloped, midway, allow upside down]{\arrow};;
    \draw [dotted] (vjb.center) -- (vjb2.center) node[sloped, midway, allow upside down]{\arrow};;
    \end{tikzpicture}
  \end{subfigure}
\caption{The exchange of goods $g_x$ and $g_y$. The edges $e_x$ and $e_y$ correspond to the respective goods in $G_{\mathcal{A}, \mathcal{B}}$, while the edges $e_x'$ and $e_y'$ correspond to those in $G_{\mathcal{A}', \mathcal{B}}$.} \label{fig:cAB}
\end{figure}

It remains to prove that the distance between $\mathcal{A}$ and $\mathcal{B}$ is at least $p_0 + 1$. 
Suppose on the contrary that the distance between $\mathcal{A}$ and $\mathcal{B}$ is at most $p_0$.
Since $c^*_{\mathcal{A}, \mathcal{B}} < m$, $\mathcal{A}$ and $\mathcal{B}$ are distinct allocations. 
Consider a shortest path between $\mathcal{A}$ and $\mathcal{B}$ on the exchange graph, and let $\mathcal{A}'$ be the allocation on this path adjacent to $\mathcal{A}$.
By assumption, the distance between $\mathcal{A}'$ and $\mathcal{B}$ is $p$ for some $p < p_0$. 
By the inductive hypothesis, $m - c^*_{\mathcal{A}', \mathcal{B}} = p$.
Let $C_{\mathcal{A}', \mathcal{B}} \in \mathfrak{C}_{\mathcal{A}', \mathcal{B}}$ be such that $|C_{\mathcal{A}', \mathcal{B}}| = m - p$; by definition of $c^*_{\mathcal{A}', \mathcal{B}}$, $C_{\mathcal{A}', \mathcal{B}}$ must consist only of directed cycles.
Now, since $\mathcal{A}$ and $\mathcal{A}'$ are adjacent on the exchange graph, there exist distinct goods $g_x$ and $g_y$ such that exchanging them in allocation $\mathcal{A}$ leads to the allocation $\mathcal{A}'$. 
Let $i_\mathcal{A}, i_\mathcal{B}, j_\mathcal{A}, j_\mathcal{B} \in N$ be such that $e_x = (i_\mathcal{A}, i_\mathcal{B})$ and $e_y = (j_\mathcal{A}, j_\mathcal{B})$ are edges in $E_{\mathcal{A}, \mathcal{B}}$ corresponding to goods $g_x$ and $g_y$, respectively (some of $i_\mathcal{A}, i_\mathcal{B}, j_\mathcal{A}, j_\mathcal{B}$ may coincide).
Accordingly, we must have edges $e'_x = (j_\mathcal{A}, i_\mathcal{B})$ and $e'_y = (i_\mathcal{A}, j_\mathcal{B})$ in $E_{\mathcal{A}', \mathcal{B}}$ (some of these edges may be self-loops).
See \Cref{fig:cAB} for an illustration.
We consider two cases; in each case, we will construct a partition $C_{\mathcal{A}, \mathcal{B}} \in \mathfrak{C}_{\mathcal{A}, \mathcal{B}}$ with at least $|C_{\mathcal{A}', \mathcal{B}}| - 1$ directed circuits.
\begin{itemize}
    \item \textbf{Case 1: $e'_x$ and $e'_y$ belong to different cycles in $C_{\mathcal{A}', \mathcal{B}}$. } \\
    Let $D_x = j_\mathcal{A} \xrightarrow{e'_x} i_\mathcal{B} \xrightarrow{\cdots} j_\mathcal{A}$ and $D_y = i_\mathcal{A} \xrightarrow{e'_y} j_\mathcal{B} \xrightarrow{\cdots} i_\mathcal{A}$ be the cycles in $C_{\mathcal{A}', \mathcal{B}}$ containing $e'_x$ and $e'_y$, respectively.
    Define $C_{\mathcal{A}, \mathcal{B}} = (C_{\mathcal{A}', \mathcal{B}} \setminus \{D_x, D_y\}) \cup \{i_\mathcal{A} \xrightarrow{e_x} i_\mathcal{B} \xrightarrow{\cdots} j_\mathcal{A} \xrightarrow{e_y} j_\mathcal{B} \xrightarrow{\cdots} i_\mathcal{A}\}$, where $i_\mathcal{B} \xrightarrow{\cdots} j_\mathcal{A}$ and $j_\mathcal{B} \xrightarrow{\cdots} i_\mathcal{A}$ represent the corresponding (possibly empty) trails in $D_x$ and $D_y$, respectively. 
    Note that $|C_{\mathcal{A}, \mathcal{B}}| = |C_{\mathcal{A}', \mathcal{B}}| - 1$.
    \item \textbf{Case 2: $e'_x$ and $e'_y$ belong to the same cycle in $C_{\mathcal{A}', \mathcal{B}}$. } \\
    Let $D = j_\mathcal{A} \xrightarrow{e'_x} i_\mathcal{B} \xrightarrow{\cdots} i_\mathcal{A} \xrightarrow{e'_y} j_\mathcal{B} \xrightarrow{\cdots} j_\mathcal{A}$ be the cycle in $C_{\mathcal{A}', \mathcal{B}}$ containing $e'_x$ and $e'_y$. 
    Define $C_{\mathcal{A}, \mathcal{B}} = (C_{\mathcal{A}', \mathcal{B}} \setminus \{D\}) \cup \{i_\mathcal{A} \xrightarrow{e_x} i_\mathcal{B} \xrightarrow{\cdots} i_\mathcal{A}, \; j_\mathcal{A}\xrightarrow{e_y} j_\mathcal{B} \xrightarrow{\cdots} j_\mathcal{A}\}$, where the $\xrightarrow{\cdots}$ represents the corresponding (possibly empty) trails in $D$. 
    Note that $|C_{\mathcal{A}, \mathcal{B}}| = |C_{\mathcal{A}', \mathcal{B}}| + 1$.
\end{itemize}
In either case, there exists a partition $C_{\mathcal{A}, \mathcal{B}} \in \mathfrak{C}_{\mathcal{A}, \mathcal{B}}$ of cardinality at least $|C_{\mathcal{A}', \mathcal{B}}| - 1 = m - p - 1$.
This means that $c^*_{\mathcal{A}, \mathcal{B}} \geq m - p - 1$, which implies that $m - c^*_{\mathcal{A}, \mathcal{B}} \leq p + 1 < p_0 + 1$.
However, this contradicts the assumption that $m - c^*_{\mathcal{A}, \mathcal{B}} = p_0 + 1$.
Therefore, the distance between $\mathcal{A}$ and $\mathcal{B}$ is at least $p_0 + 1$, as desired.
\end{proof}

Having found a correspondence for the distance between two allocations, a natural question is whether there exists an efficient algorithm to compute this distance.
It turns out that computing this distance is an NP-hard problem, so no polynomial-time algorithm exists unless $\text{P} = \text{NP}$.
We show this via a series of reductions.

We start by considering the decision problem \textsc{Directed Triangle Partition}: given a directed graph with no directed cycles of length $1$ or $2$, determine whether there is a partition of edges into triangles (i.e., directed cycles of length $3$).
This decision problem is NP-hard via a reduction from \textsc{3SAT}.
The idea is similar to that used by \citet{Holyer81} in his proof of the corresponding result for \emph{undirected} graphs; the details are involved and can be found in \Cref{app:directedtrianglepartition}.

\begin{restatable}{lemma}{directedtrianglepartitionnphard} \label{lem:directed_triangle_partition_nphard}
\textsc{Directed Triangle Partition} is NP-hard.
\end{restatable}

We now use \Cref{lem:directed_triangle_partition_nphard} to show that computing $c^*_{\mathcal{A}, \mathcal{B}}$ is NP-hard.

\begin{lemma} \label{lem:max_partition_nphard}
Given a directed graph such that for each vertex, its indegree and outdegree are equal, computing the maximum cardinality of a partition of the edges into directed cycles is an NP-hard problem.
\end{lemma}

\begin{proof} 
The result follows from reducing \textsc{Directed Triangle Partition} to the problem of deciding whether there exists a partition of the edges of a directed graph into $|E|/3$ directed cycles.
Let $G = (V, E)$ be an instance of \textsc{Directed Triangle Partition}.
If there is some vertex with unequal indegree and outdegree, then $G$ cannot be edge-partitioned into triangles. 
Otherwise, since $G$ does not have cycles of length $1$ or $2$ (by definition of \textsc{Directed Triangle Partition}), the edges of $G$ can be partitioned into triangles if and only if the maximum cardinality of a partition of the edges into directed cycles is $|E|/3$. 
Since \textsc{Directed Triangle Partition} is NP-hard by \Cref{lem:directed_triangle_partition_nphard}, so is the problem of finding the maximum cardinality of a partition of the edges into directed cycles.
\end{proof}

\Cref{lem:distance_between_allocations} and \Cref{lem:max_partition_nphard} imply the following theorem.

\begin{theorem} \label{thm:dist_nphard}
Finding the distance between two allocations is an NP-hard problem.
\end{theorem}

\begin{proof}
Start with an instance $\mathcal{G} = (V, E)$ of the problem described in \Cref{lem:max_partition_nphard}, and denote $V = \{v_1, \ldots, v_n\}$. 
We shall construct, in polynomial time, an instance of the problem of finding the distance between two allocations.
Let $N = \{1, \ldots, n\}$ be the set of agents, $M = \{g_e\}_{e \in E}$ be the set of goods, and $s_i = |\{e \in E \mid \exists j \in N, \; e = (v_i, v_j)\}|$ be the size of agent $i$'s bundle for each $i \in N$. 
The initial allocation $\mathcal{A} = (A_1, \ldots, A_n)$ and target allocation $\mathcal{B} = (B_1, \ldots, B_n)$ are such that $A_i = \{g_e \in M \mid \exists j \in N, \; e = (v_i, v_j) \in E\}$ and $B_i = \{g_e \in M \mid \exists j \in N, \; e = (v_j, v_i) \in E\}$ for each $i \in N$.
Note that this induces the graph $G_{\mathcal{A}, \mathcal{B}}$ isomorphic to $\mathcal{G}$. 
The distance between $\mathcal{A}$ and $\mathcal{B}$ is $|E| - c^*_{\mathcal{A}, \mathcal{B}}$, by \Cref{lem:distance_between_allocations}.
Therefore, if we can find this distance, then we can find $c^*_{\mathcal{A}, \mathcal{B}}$, solving the problem instance from \Cref{lem:max_partition_nphard}.
\end{proof}

\subsection{General Utilities}

We now discuss properties of the EF1 exchange graph.
The following result demonstrates that deciding whether an EF1 exchange path exists is a PSPACE-complete problem.

\begin{theorem} \label{thm:gen_2_connected_pspace}
Deciding the existence of an EF1 exchange path between two EF1 allocations is PSPACE-complete.
\end{theorem}

\begin{proof}
First, we show membership in PSPACE---recall that PSPACE is the set of all decision problems that can be solved by a deterministic polynomial-space Turing machine.
We can solve the problem non-deterministically by simply guessing an EF1 exchange path between the two EF1 allocations.
Since the total number of allocations is at most $n^m$, if there exists an EF1 exchange path between the two allocations, then there exists one with length at most $n^m$; such a path can be verified in polynomial space (i.e., using a polynomial number of bits).
This shows that the problem is in NPSPACE, the set of all decision problems that can be solved by a non-deterministic polynomial-space Turing machine.
By Savitch's Theorem, $\text{NPSPACE} \subseteq \text{PSPACE}$ \citep{Savitch70}, which implies that this problem is in PSPACE.

To prove that our problem is PSPACE-hard, we shall reduce the \textsc{Perfect Matching Reconfiguration} problem for a balanced (undirected) bipartite graph to our problem.
Recall that the \textsc{Perfect Matching Reconfiguration} problem is the task of deciding if two perfect matchings of a balanced bipartite graph can be reached from each other via a sequence of \emph{flips}, i.e., given perfect matchings $\widetilde{W}_0$ and $\widetilde{W}$ of a balanced bipartite graph $\widetilde{G} = (\widetilde{V}, \widetilde{E})$, whether there exists a sequence of perfect matchings $\widetilde{W}_0, \widetilde{W}_1, \ldots, \widetilde{W}_t$ such that
\begin{itemize}
    \item $\widetilde{W}_t = \widetilde{W}$, and
    \item for each $z \in \{1, \ldots, t\}$, there exist edges $\widetilde{e}_z^1, \widetilde{e}_z^2, \widetilde{e}_z^3, \widetilde{e}_z^4$ of $\widetilde{G}$ such that $\widetilde{W}_{z-1} \setminus \widetilde{W}_z = \{\widetilde{e}_z^1, \widetilde{e}_z^3\}$, $\widetilde{W}_z \setminus \widetilde{W}_{z-1} = \{\widetilde{e}_z^2, \widetilde{e}_z^4\}$, and $\widetilde{e}_z^1 \widetilde{e}_z^2 \widetilde{e}_z^3 \widetilde{e}_z^4$ forms a cycle.
\end{itemize}
The operation from $\widetilde{W}_{z-1}$ to $\widetilde{W}_z$ is called a \emph{flip}, and we say that $\widetilde{W}_{z-1}$ and $\widetilde{W}_z$ are \emph{adjacent} to each other. 
\textsc{Perfect Matching Reconfiguration} is known to be PSPACE-hard \citep{BonamyBoHe19}. 
Let $|\widetilde{V}| = 2v$, and let the two independent sets of~$\widetilde{G}$ be $\widetilde{P} = \{\widetilde{p}_1, \ldots, \widetilde{p}_v\}$ and $\widetilde{Q} = \{\widetilde{q}_1, \ldots, \widetilde{q}_v\}$.
For each $i \in \{1, \ldots, v\}$, let $\widetilde{q}_{k_i} \in \widetilde{Q}$ be the vertex adjacent to $\widetilde{p}_i$ in $\widetilde{W}_0$, and let $\widetilde{q}_{\ell_i} \in \widetilde{Q}$ be the vertex adjacent to $\widetilde{p}_i$ in $\widetilde{W}$. 
We shall show that this problem instance can be reduced to an instance of deciding the existence of an EF1 exchange path between two EF1 allocations.

Define an instance of the EF1 exchange path problem as follows: let $N = \{0, 1, \ldots, v\}$ be the set of agents, $M = \{p_1, \ldots, p_v, q_1, \ldots, q_v, r_1, r_2, r_3, r_4\}$ be the set of goods, and the utility function of each agent be 
\begin{itemize}
    \item $u_0(g) = 0$ for all $g \in M$, and
    \item for $i \in \{1, \ldots, v\}$, 
    \begin{align*}
        u_i(g) = \left\{ \begin{array}{ll}
            3 & \text{if } g \in \{p_i\} \cup \{q_k \; | \; \{\widetilde{p}_i, \widetilde{q}_k \} \in \widetilde{E} \}; \\
            2 & \text{if } g \in \{r_1, r_2, r_3, r_4\}; \\
            0 & \text{otherwise.}
        \end{array} \right.
    \end{align*}
\end{itemize}
In the initial allocation $\mathcal{A}_0$, agent $0$ has the bundle $\{r_1, r_2, r_3, r_4\}$ and agent $i$ has the bundle $\{p_i, q_{k_i}\}$ for each $i \in \{1, \ldots, v\}$. 
In the target allocation $\mathcal{A}$, agent $0$ again has the bundle $\{r_1, r_2, r_3, r_4\}$ and agent $i$ has the bundle $\{p_i, q_{\ell_i}\}$ for each $i \in \{1, \ldots, v\}$. 
Observe that both allocations are EF1---agent $0$ assigns zero utility to every bundle, while each agent $i \in \{1, \ldots, v\}$ assigns a utility of $6$ to her own bundle, a utility of at most $6$ to the bundle of every agent in $\{1, \ldots, v\} \setminus \{i\}$, and a utility of $6 + 2$ to agent $0$'s bundle. 
Clearly, this instance can be constructed in polynomial time.

Suppose first that there exists a sequence of adjacent perfect matchings from $\widetilde{W}_0$ to~$\widetilde{W}$.
Then each flip from $\widetilde{W}_{z-1}$ to $\widetilde{W}_z$ corresponds to an exchange in the EF1 exchange path problem: if $\widetilde{W}_{z-1} \setminus \widetilde{W}_z = \{\{\widetilde{p}_i, \widetilde{q}_k\}, \{\widetilde{p}_j, \widetilde{q}_\ell\}\}$ and $\widetilde{W}_z \setminus \widetilde{W}_{z-1} = \{\{\widetilde{p}_i, \widetilde{q}_\ell\}, \{\widetilde{p}_j, \widetilde{q}_k\}\}$, then exchange $q_k$ in agent $i$'s bundle with $q_\ell$ in agent $j$'s bundle. 
The new allocation is also EF1---as before, agent~$0$ assigns zero utility to every bundle, while each agent $i' \in \{1, \ldots, v\}$ assigns a utility of $6$ to her own bundle, a utility of at most $6$ to the bundle of every agent in $\{1, \ldots, v\} \setminus \{i'\}$, and a utility of $6 + 2$ to agent $0$'s bundle. 
By performing the exchanges according to the flips in sequence, we reach the target allocation. Therefore, an EF1 exchange path exists.

Conversely, assume that an EF1 exchange path exists between the initial allocation $\mathcal{A}_0$ and the target allocation $\mathcal{A}$.
Consider the sequence of EF1 allocations $\mathcal{A}_0, \mathcal{A}_1, \ldots, \mathcal{A}_t = \mathcal{A}$. 
We show by induction that for every intermediate allocation $\mathcal{A}_z$, every agent $i \in \{1, \ldots, v\}$ assigns a utility of $6$ to her own bundle (consisting of $p_i$ and $q_k$ for some $k$), and agent $0$ retains $\{r_1, r_2, r_3, r_4\}$. 
The base case $z = 0$ is trivial. 
For the inductive case, suppose that the statement is true for $z-1$. 
If some agent $i \in \{1, \ldots, v\}$ attempts to exchange one of her goods with a good from agent $0$, then agent $i$'s new bundle has utility $5$ but agent~$0$'s new bundle has utility $6 + 3$ for agent~$i$, which violates EF1. 
Therefore, agent $i$ must exchange goods with agent $j$ for some $j \in \{1, \ldots, v\}$. 
Note that agent $0$'s bundle is worth $6 + 2$ to agent $i$ and $j$, so agent $i$'s and $j$'s own bundles must be worth at least $6$ to $i$ and $j$ respectively. 
If agent $i$ gives $p_i$ to agent $j$, then agent~$j$'s new bundle consists of $p_i$ (worth zero to her) and some $q_\ell$, which is worth at most $3$ to her---this violates EF1. 
By the same reasoning, agent $j$ cannot give $p_j$ to agent $i$. 
Therefore, they must exchange $q_k$ in agent $i$'s bundle with $q_\ell$ in agent $j$'s bundle. As agent $i$ and $j$ must have bundles worth at least $6$ to each of them, $q_\ell$ must be worth $3$ to agent $i$ and $q_k$ must be worth $3$ to agent $j$. 
This completes the induction. 

As a result, the perfect matchings $\widetilde{W}_{z-1}$ (corresponding to $\mathcal{A}_{z-1}$) and $\widetilde{W}_z$ (corresponding to $\mathcal{A}_z$) must be adjacent to each other for all $z$, where $\widetilde{W}_{z-1} \setminus \widetilde{W}_z = \{\{\widetilde{p}_i, \widetilde{q}_k\}, \{\widetilde{p}_j, \widetilde{q}_\ell\}\}$ and $\widetilde{W}_z \setminus \widetilde{W}_{z-1} = \{\{\widetilde{p}_i, \widetilde{q}_\ell\}, \{\widetilde{p}_j, \widetilde{q}_k\}\}$.
It follows that a sequence of adjacent perfect matchings $\widetilde{W}_0, \widetilde{W}_1, \ldots, \widetilde{W}_t = \widetilde{W}$ indeed exists.
This completes the proof.
\end{proof}

Regarding the existence of \emph{optimal} EF1 exchange paths, we shall show later in \Cref{thm:iden_4_optimal_nphard} that the corresponding decision problem is NP-hard even for four agents with identical utilities.

\subsection{Identical Binary Utilities}
We now consider the most restrictive class of utility functions in our paper: those that are identical \emph{and} binary.
We show that the EF1 exchange graph is connected for any number of agents with such utility functions.

\begin{theorem} \label{thm:idenbin_2_connected}
Let an instance with $n \geq 2$ agents and identical binary utility functions be given. Then, the EF1 exchange graph is connected. Moreover, an EF1 exchange path between any two allocations can be found in polynomial time.
\end{theorem}

\begin{proof}
Let $\mathcal{A}$ and $\mathcal{B}$ be two EF1 allocations. 
Since every good is worth either $1$ or $0$ to every agent, every agent's bundle in $\mathcal{A}$ and $\mathcal{B}$ must have a utility of either $\lfloor u(M)/n \rfloor$ or $\lfloor u(M)/n \rfloor + 1$ (otherwise, the gap between the utilities of some two agents' bundles is at least $2$, and the corresponding allocation is not EF1). 
Let $N'$ be the set of agents whose bundles in $\mathcal{A}$ and $\mathcal{B}$ have different utilities. 
Note that half of the agents in $N'$ have bundles worth $\lfloor u(M)/n \rfloor$ in $\mathcal{A}$ and $\lfloor u(M)/n \rfloor + 1$ in $\mathcal{B}$; the other half have bundles worth $\lfloor u(M)/n \rfloor + 1$ in $\mathcal{A}$ and $\lfloor u(M)/n \rfloor$ in $\mathcal{B}$. 
If $N'\ne\emptyset$, let agent $i \in N'$ be an agent with a bundle worth $\lfloor u(M)/n \rfloor$ in $\mathcal{A}$, and let $g_i$ be a good with utility $0$ in $A_i$---this good exists because agent $i$ has at least $\lfloor u(M)/n \rfloor + 1$ goods in her bundle (due to $B_i$'s utility of $\lfloor u(M)/n \rfloor + 1$) but only has utility $\lfloor u(M)/n \rfloor$ in $A_i$. 
Let agent $j \in N'$ be an agent with a bundle worth $\lfloor u(M)/n \rfloor + 1$ in $\mathcal{A}$, and let $g_j$ be a good with utility $1$ in $A_j$---this good exists because $A_j$ has utility at least $1$. 
Exchange $g_i$ with $g_j$; it can be verified that the resulting allocation is EF1. 
As this exchange reduces the size of the set $N'$ by two, we can repeatedly make such exchanges between two agents in $N'$ until $N' = \emptyset$. 
Note that such exchanges can be performed in polynomial time.

At this point, we have shown that there exists an EF1 allocation $\mathcal{A}'$ such that an EF1 exchange path exists between $\mathcal{A}$ and $\mathcal{A}'$, and for every agent $i$, her bundles in $\mathcal{A}'$ and~$\mathcal{B}$ have the same utility. 
Define the item graph $G_{\mathcal{A}', \mathcal{B}}$ as in the beginning of \Cref{sec:three_agents}, and consider its subgraph with only the edges representing the goods with utility $1$. 
For each agent, the indegree and the outdegree of the corresponding vertex in this subgraph are equal, so we can perform exchanges to `resolve' these edges.
Specifically, suppose there is an edge $e_x = (i, j)$ corresponding to a good $g_x$, where $i \ne j$.
By the degree condition, there must exist another edge $e_y = (j, k)$ corresponding to a good $g_y$, where $j\ne k$ but possibly $k = i$.
We let agents $i$ and $j$ exchange $g_x$ and $g_y$, so $g_x$ is now with its correct owner, agent~$j$.
Hence, at least one more good goes to the correct agent after the exchange.
This exchange process can be performed in polynomial time, and no agent's utility changes during the process, which means that the intermediate allocations are all EF1.
Similarly, if we consider the subgraph with only the edges representing the goods with utility $0$, we can perform exchanges to resolve these edges as well. 
Therefore, there exists an EF1 exchange path from $\mathcal{A}'$ to $\mathcal{B}$, and thus an EF1 exchange path from $\mathcal{A}$ to~$\mathcal{B}$, and this path can be found in polynomial time.
\end{proof}

In spite of this positive result, the polynomial-time algorithm described in \Cref{thm:idenbin_2_connected} does not necessarily find an optimal EF1 exchange path between the two allocations.
In fact, even for the special case where the EF1 exchange graph $H$ and the exchange graph~$G$ coincide (e.g., when every agent assigns zero utility to every good, so every allocation is EF1), it is NP-hard to compute an optimal EF1 exchange path by \Cref{thm:dist_nphard}, regardless of whether optimality refers to the length of the shortest path in $G$ or in~$H$.
Hence, a polynomial-time algorithm for this task does not exist unless $\text{P} = \text{NP}$.
Moreover, we show next that, somewhat surprisingly, an optimal EF1 exchange path (with respect to $G$) is not guaranteed to exist even for identical binary utilities.

\begin{theorem} \label{thm:idenbin_3_optimal}
For each $n \geq 3$, there exists an instance with $n$ agents with identical binary utility functions satisfying the following properties: the EF1 exchange graph is connected, but for some pair of EF1 allocations, no optimal EF1 exchange path exists between them.
\end{theorem}

\begin{proof}
For $n = 3$ agents, consider $\vec{s} = (2, 2, 2)$ and the utility of the goods as follows:
\begin{center}
\begin{tabular}{c|cccccc}
$g$    & $g_1$ & $g_2$ & $g_3$ & $g_4$ & $g_5$ & $g_6$ \\ \hline
$u(g)$ & $1$   & $1$   & $1$   & $0$   & $0$   & $0$   \\ 
\end{tabular}
\end{center}
Note that the EF1 exchange graph is connected by \Cref{thm:idenbin_2_connected}. 
We prove that an optimal EF1 exchange path between $\mathcal{A}$ and $\mathcal{B}$ does not exist, where $A_1 = \{g_2, g_6\}$, $A_2 = \{g_3, g_4\}$, $A_3 = \{g_1, g_5\}$, and $B_i = \{g_i, g_{i+3}\}$ for $i \in \{1, 2, 3\}$---it can be verified that both $\mathcal{A}$ and $\mathcal{B}$ are EF1, and the distance between $\mathcal{A}$ and $\mathcal{B}$ is $3$ (through exchanging $g_1 \leftrightarrow g_6$, $g_2 \leftrightarrow g_4$, and $g_3 \leftrightarrow g_5$). 
Consider any EF1 exchange path between $\mathcal{A}$ and $\mathcal{B}$, and let $\mathcal{A}'$ be the EF1 allocation adjacent to $\mathcal{A}$ on the exchange path. 
If a valuable good ($g_1$, $g_2$, or $g_3$) is exchanged with a non-valuable good ($g_4$, $g_5$, or $g_6$) from $\mathcal{A}$ to reach~$\mathcal{A}'$, then one agent has utility $0$ and another agent has utility $2$, which means that $\mathcal{A}'$ is not EF1. 
Therefore, the only exchanges possible from $\mathcal{A}$ are between valuable goods or between non-valuable goods. 
However, any of these exchanges causes at most one good to go to the correct bundle according to $\mathcal{B}$, so there are at least five goods in $\mathcal{A}'$ in the wrong bundle according to $\mathcal{B}$. 
As any exchange of goods reduces the number of goods in the wrong bundle by at most two, the distance between $\mathcal{A}'$ and $\mathcal{B}$ is at least $3$. 
This means that the distance between $\mathcal{A}$ and $\mathcal{B}$ is at least $4$.
It follows that no optimal EF1 exchange path exists between $\mathcal{A}$ and $\mathcal{B}$.

For $n > 3$ agents, simply add $n-3$ dummy agents who have the same utility function as the three original agents and have empty bundles.
\end{proof}

\subsection{Binary Utilities}
We saw in \Cref{thm:idenbin_2_connected} that the EF1 exchange graph is always connected for any number of agents with identical binary utilities.
Now, we consider the case where the agents have binary utilities which may differ between agents.
It turns out that the EF1 exchange graph is not necessarily connected in this case, even when there are three agents.
This also provides a contrast to the case of two agents (\Cref{thm:binary_2_optimal}).

\begin{theorem} \label{thm:binary_3_connected}
For each $n \geq 3$, there exists an instance with $n$ agents with binary utility functions such that the EF1 exchange graph is disconnected.
\end{theorem}

\begin{proof}
For $n = 3$ agents, consider the utility of the goods as follows:
\begin{center}
\begin{tabular}{c|cccccc}
$g$      & $g_1$ & $g_2$ & $g_3$ & $g_4$ \\ \hline
$u_1(g)$ & $1$   & $0$   & $1$   & $0$   \\
$u_2(g)$ & $1$   & $0$   & $1$   & $0$   \\
$u_3(g)$ & $0$   & $1$   & $1$   & $0$
\end{tabular}
\end{center}
Let $\mathcal{A}$ and $\mathcal{B}$ be given such that $A_1 = \{g_1, g_2\}$, $A_2 = \{g_3, g_4\}$, $B_1 = \{g_3, g_4\}$, $B_2 = \{g_1, g_2\}$, and $A_3 = B_3 = \emptyset$---it can be verified that both $\mathcal{A}$ and $\mathcal{B}$ are EF1.
Consider any EF1 exchange path between $\mathcal{A}$ and $\mathcal{B}$, and let $\mathcal{A}'$ be the EF1 allocation adjacent to~$\mathcal{A}$ on the exchange path. 
We claim that $\mathcal{A}'$ cannot exist. 
The only possible exchanges from~$\mathcal{A}$ are $g_i \leftrightarrow g_{5-i}$ or $g_i \leftrightarrow g_{i+2}$ for some $i \in \{1, 2\}$. 
If $g_i$ is exchanged with $g_{5-i}$, then agent $i$'s bundle has zero utility from agent $i$'s perspective while agent $(3-i)$'s bundle has utility~$2$ from agent $i$'s perspective, so agent $i$ envies agent $(3-i)$ by more than one good. 
On the other hand, if $g_i$ is exchanged with $g_{i+2}$, then agent $i$'s bundle has utility $2$ from agent~$3$'s perspective, and since agent $3$ has an empty bundle, agent $3$ envies agent $i$ by more than one good. 
Therefore, $\mathcal{A}'$ does not exist, which contradicts the assumption that the path is an EF1 exchange path. 
It follows that no EF1 exchange path exists between $\mathcal{A}$ and $\mathcal{B}$.

For $n > 3$ agents, simply add dummy agents who assign zero value to every good and have empty bundles.
\end{proof}

\subsection{Identical Utilities} \label{sec:iden_util}
Let us now consider the case where the utilities are identical across agents, though they need not be binary. 
As with the case of binary utilities, there are instances in which the EF1 exchange graph is not connected even for three agents.

\begin{theorem} \label{thm:iden_3_connected}
For each $n \geq 3$, there exists an instance with $n$ agents with identical utility functions such that the EF1 exchange graph is disconnected.
\end{theorem}

\begin{proof}
For $n = 3$ agents, consider the utility of the goods as follows:
\begin{center}
\begin{tabular}{c|ccccccc}
$g$    & $g_1$ & $g_2$ & $g_3$ & $g_4$ & $g_5$ & $g_6$ & $g_7$ \\ \hline
$u(g)$ & $4$   & $3$   & $1$   & $4$   & $2$   & $2$   & $4$  
\end{tabular}
\end{center}
Let $\mathcal{A}$ and $\mathcal{B}$ be given such that $A_1 = \{g_1, g_2, g_3\}$, $A_2 = \{g_4, g_5, g_6\}$, $B_1 = \{g_1, g_5, g_6\}$, $B_2 = \{g_2, g_3, g_4\}$, and $A_3 = B_3 = \{g_7\}$.
It can be verified that both $\mathcal{A}$ and $\mathcal{B}$ are EF1.
Consider any exchange path from $\mathcal{A}$ to $\mathcal{B}$.
At some point, a good $g \in \{g_2, g_3, g_5, g_6\}$ has to be exchanged, but this will inevitably cause agent $3$ to envy agent $1$ or agent $2$ by more than one good, so the exchange path cannot be EF1. 
Therefore, no EF1 exchange path exists between $\mathcal{A}$ and $\mathcal{B}$. 

For $n > 3$ agents, simply add $n - 3$ agents who have the same utility function as the three original agents and $n - 3$ goods with value~$4$ each, and allocate each of these goods to one of these $n - 3$ agents in both $\mathcal{A}$ and $\mathcal{B}$.
\end{proof}

We end this section with a result that determining whether an optimal EF1 exchange path exists between two allocations is NP-hard even for four agents with identical valuations. 
This can be shown via a reduction from the NP-hard problem \textsc{Partition}.

\begin{theorem} \label{thm:iden_4_optimal_nphard}
Deciding the existence of an optimal EF1 exchange path between two EF1 allocations is NP-hard, even for $n = 4$ agents with identical utility functions.
\end{theorem}

\begin{proof}
We shall reduce the NP-hard problem \textsc{Partition} to this problem. 
Recall that the \textsc{Partition} problem is the task of deciding whether a multiset $T = \{t_1, \ldots, t_k\}$ of positive integers can be partitioned into two subsets such that the sum of the integers in one subset is equal to that in the other subset. 
Let the sum of all the integers in $T$ be $2S$.

Define an instance of the EF1 exchange path problem with $n = 4$ agents and the set of goods $M = \{a_0, a_1, \ldots, a_k, b_0, b_1, \ldots, b_k, c_1, c_2, d_1, d_2\}$. 
The utility of each good is defined as follows:
\begin{itemize}
    \item $u(a_0) = u(b_0) = u(c_1) = 2S$,
    \item $u(d_1) = u(d_2) = S$,
    \item $u(a_i) = t_i$ for all $i \in \{1, \ldots, k\}$,
    \item $u(b_i) = u(c_2) = 0$ for all $i \in \{1, \ldots, k\}$.
\end{itemize}

The initial allocation $\mathcal{A} = (A_1, A_2, A_3, A_4)$ and the target allocation $\mathcal{B} = (B_1, B_2, B_3, B_4)$ are given by $A_1 = \{a_0, a_1, \ldots, a_k\}$, $A_2 = \{b_0, b_1, \ldots, b_k\}$, $B_1 = \{a_0, b_1, \ldots, b_k\}$, $B_2 = \{b_0, a_1, \ldots, a_k\}$, $A_3 = B_4 = \{c_1, c_2\}$, and $A_4 = B_3 = \{d_1, d_2\}$---it can be verified that both $\mathcal{A}$ and $\mathcal{B}$ are EF1. 
Note that this instance can be constructed in polynomial time, and the distance between $\mathcal{A}$ and $\mathcal{B}$ is $k + 2$.

First, suppose that $T$ can be partitioned into two subsets of sum $S$ each, say $\{t_{i_1}, \ldots, t_{i_\ell}\}$ has sum $S$.
We perform the following exchanges starting from $\mathcal{A}$:
\begin{itemize}
    \item First, exchange $\{a_{i_1}, \ldots, a_{i_\ell}\}$ with $\{b_{i_1}, \ldots, b_{i_\ell}\}$ pair-by-pair. 
    At this point, agents $1$ and $2$ have bundles worth $3S$ each, and agents $3$ and $4$ have bundles worth $2S$ each.
    \item Next, exchange $c_1$ with $d_1$, and exchange $c_2$ with $d_2$. 
    \item Finally, exchange $\{a_1, \ldots, a_k\} \setminus \{a_{i_1}, \ldots, a_{i_\ell}\}$ with $\{b_1, \ldots, b_k\} \setminus \{b_{i_1}, \ldots, b_{i_\ell}\}$ pair-by-pair.
\end{itemize}
The allocation resulting from this sequence of exchanges is $\mathcal{B}$. 
It can be verified that this exchange path has length $k + 2$ and every intermediate allocation is EF1.
Hence, there exists an optimal EF1 exchange path between $\mathcal{A}$ and $\mathcal{B}$.

Conversely, suppose that there exists an optimal EF1 exchange path between $\mathcal{A}$ and~$\mathcal{B}$.
The only exchanges possible are $a_i \leftrightarrow b_j$ for some $i, j \in \{1, \ldots, k\}$ and $c_i \leftrightarrow d_j$ for some $i, j \in \{1, 2\}$. 
In particular, at some point, $c_i$ must be exchanged with $d_j$ for the first time for some $i, j \in \{1, 2\}$. 
Consider the allocation following this exchange.
One of agents $3$ and $4$ now has utility only $S$.
By assumption, this allocation is EF1, so this agent does not envy agent $1$ and agent $2$.
Removing the highest-utility good from agent~$1$'s and agent $2$'s bundle (i.e., $a_0$ and $b_0$), the utility of each of the remaining bundles must be at most $S$. 
The only way this is possible is that $\{a_1, \ldots, a_k\}$ can be partitioned into two subsets such that the utility of each subset is exactly $S$, and each of agents $1$ and $2$ receives exactly one of those subsets. 
Correspondingly, this shows that $T$ can be partitioned into two subsets of sum $S$ each.
\end{proof}

\section{Conclusion and Future Work}

In this paper, we have initiated the study of reachability problems in fair division by investigating the connectivity of the EF1 exchange graph and the optimality of EF1 exchange paths. 
We showed that even for two agents, an EF1 exchange path between two given EF1 allocations does not necessarily exist.
On the positive side, such a path always exists if both agents have identical or binary utility functions---in these cases, we can also ensure an optimal path regardless of EF1 considerations, and the path can be found in polynomial time.
For three or more agents, however, the problem becomes much less tractable, both in terms of existence and computation.
In particular, we proved that finding the smallest number of exchanges between two allocations is NP-hard even if we were to ignore the EF1 constraints, and deciding whether an EF1 exchange path between two allocations exists is PSPACE-complete.
Moreover, the existence of an EF1 exchange path cannot be guaranteed even if the utilities are identical \emph{or} binary, although such a guarantee is possible if the utilities are identical \emph{and} binary.

Our work leaves several questions and directions for future research.
Firstly, while determining the existence of an EF1 exchange path between two given allocations is PSPACE-complete in general, an intriguing question is whether this can be done in polynomial time for two agents. 
In addition, for the negative results obtained in our paper, one could ask whether an (optimal or otherwise) exchange path between EF1 allocations exists if we allow the intermediate allocations to be \emph{envy-free up to $k$ goods (EF$k$)} for some small $k > 1$.
Extending our results to fairness notions other than EF1---for example, the share-based notion \emph{maximin share fairness (MMS)} \citep{Budish11,KurokawaPrWa18}---is also a meaningful direction.
Finally, in addition to (or instead of) exchanging goods between agents, one may also consider the setting where an agent \emph{transfers} one good to another agent in each operation---in this case, the size of the allocation does not need to be fixed.
In \Cref{app:transfer}, we present some results concerning this setting, which differs from the exchange-only setting in interesting ways.

\section*{Acknowledgments}

This work was partially supported by the Singapore Ministry of Education under grant
number MOE-T2EP20221-0001, by JST PRESTO under grant number JPMJPR20C1, by JSPS KAKENHI under grant number JP20H05795, by JST ERATO under grant number JPMJER2301, and by an NUS Start-up Grant.
We thank the anonymous AAAI 2024 and Algorithmica reviewers for their valuable comments.

\bibliographystyle{plainnat}
\bibliography{refs}

\appendix

\section{Transfers}
\label{app:transfer}

Thus far, we have considered the operation where two agents exchange a pair of goods with each other. 
In this appendix, we examine the scenario where in each operation, an agent is allowed to \emph{transfer} a good from her bundle to another agent's bundle without receiving a good from the latter agent in return. 
This allows more flexibility in the allocation size vector, as it is now possible for this vector to change after each (transfer) operation.

First, we consider the setting where the agents are allowed to transfer goods one at a time, but \emph{not} allowed to exchange goods. We define \emph{(EF1) transfer graph} and \emph{(EF1) transfer path} analogously to the corresponding exchange definitions in \Cref{sec:prelim}. Note that for transfer graphs, the set of vertices consists of \emph{all} allocations, rather than only those with a certain size vector. It turns out that this setting is even more restricted than the exchange-only setting.

\begin{theorem} 
For each $n \geq 2$, there exists an instance with $n$ agents with identical binary utility functions such that the EF1 transfer graph is disconnected.
\end{theorem}

\begin{proof}
For $n$ agents, consider $2n$ goods, each with utility $1$ to every agent. 
In the initial allocation each agent has two goods each, while in the target allocation agent $1$'s and~$2$'s bundles are exchanged from the initial allocation. 
As the only operation available is a transfer operation, some agent will have three goods and another agent will have only one good after a transfer, so the allocation cannot be EF1.
\end{proof}

In light of the above negative result, it is more interesting to allow exchange operations on top of transfer operations. We define \emph{(EF1) exchange-and-transfer graph} and \emph{(EF1) exchange-and-transfer path} analogously to the transfer versions. Likewise, the set of vertices consists of \emph{all} allocations, rather than only those with a certain size vector.

Having the additional transfer operation on top of the exchange operation allows more flexibility in reaching other allocations. 
In general, a transfer in the exchange-and-transfer setting is equivalent to supplementing every agent with sufficiently many dummy goods (with zero utility) in the exchange-only setting and exchanging a normal good with a dummy good.
This allows us to achieve the same positive results for connectivity as in the exchange-only setting for two agents with identical or binary utilities, and for any number of agents with identical binary utilities (Theorems \ref{thm:iden_2_optimal}, \ref{thm:binary_2_optimal}, and \ref{thm:idenbin_2_connected}).

\begin{theorem} 
Let an instance with $n = 2$ agents and identical utilities be given. 
Then, the EF1 exchange-and-transfer graph is connected.
Moreover, an EF1 exchange-and-transfer path between any two allocations can be found in polynomial time.
\end{theorem}

\begin{theorem}
Let an instance with $n = 2$ agents and binary utilities be given. 
Then, the EF1 exchange-and-transfer graph is connected.
Moreover, an EF1 exchange-and-transfer path between any two allocations can be found in polynomial time.
\end{theorem}

\begin{theorem} \label{thm:idenbin_2_exchangetransfer}
Let an instance with $n \geq 2$ agents and identical binary utilities be given.
Then, the EF1 exchange-and-transfer graph is connected. 
Moreover, an EF1 exchange-and-transfer path between any two allocations can be found in polynomial time.
\end{theorem}

We shall prove \Cref{thm:idenbin_2_exchangetransfer}; the other two results can be shown in a similar manner.

\begin{proof}[Proof of \Cref{thm:idenbin_2_exchangetransfer}]
Let an instance in the exchange-and-transfer setting be given and let $\mathcal{A}$ and $\mathcal{B}$ be the initial and target EF1 allocations, respectively.
Consider the following instance in the exchange-only setting: augment the set of goods with $(n-1)m$ dummy goods yielding zero utility to all agents, and let the size vector be $(s_1', \ldots, s_n') = (m, \ldots, m)$, where $m$ denotes the number of goods in the \emph{original} instance. 
For the initial and target allocations $\mathcal{A}'$ and $\mathcal{B}'$ in the exchange-only setting, each agent receives the same set of goods as she had in $\mathcal{A}$ and~$\mathcal{B}$, respectively, together with a number of dummy goods so that her bundle has exactly $m$ goods.
Note that the utility functions remain identical and binary in the new instance.

By \Cref{thm:idenbin_2_connected}, there exists an EF1 exchange path between $\mathcal{A}'$ and $\mathcal{B}'$. 
Take one such EF1 exchange path, and consider each operation. 
If two non-dummy goods are exchanged, we exchange the same two goods in the original exchange-and-transfer instance. 
If one non-dummy good and one dummy good are exchanged, we transfer  the non-dummy good in the original instance to the other agent involved in the exchange. 
Else, if two dummy goods are exchanged, we do nothing. 
Since the dummy goods do not play any role in determining whether an allocation is EF1 in the new instance, there exists an EF1 exchange-and-transfer path between $\mathcal{A}$ and $\mathcal{B}$ in the original instance. 
The polynomial time claim follows because, by \Cref{thm:idenbin_2_connected}, an EF1 exchange path in the new instance can be found in polynomial time.
\end{proof}

In addition to admitting analogous results, the exchange-and-transfer setting sometimes allows us to derive positive results which cannot be attained in the exchange-only setting.
For instance, we show that the counterexample in \Cref{thm:iden_3_connected} is no longer applicable in the exchange-and-transfer context. 
Recall that the counterexample is for three agents with identical utilities having goods with the following utilities:
\begin{center}
\begin{tabular}{llll}
Agent $1$: & $4$, & $3$, & $1$ \\
Agent $2$: & $4$, & $2$, & $2$ \\
Agent $3$: & $4$  &      &
\end{tabular}
\end{center}
and the target allocation is to exchange $\{3, 1\}$ with $\{2, 2\}$, where we refer to a good by its utility. 
First, the good with utility $3$ can be \emph{transferred} from agent~$1$ to agent $3$. Next, the goods with utilities $1$ and $2$ can be exchanged. After that, the remaining good with utility $2$ in agent $2$'s bundle can be transferred to agent $1$. Finally, the good with utility $3$ in agent $3$'s bundle can be transferred to agent $2$.

In fact, we now state and prove a positive result which is more general than this observation.

\begin{theorem} \label{thm:exchange_and_transfer}
Suppose three agents with identical utilities have goods with the following utilities:
\begin{center}
\begin{tabular}{llllll}
Agent $1$: & $a_0$, & $a_1$, & $a_2$, & $\ldots$, & $a_k$ \\
Agent $2$: & $b_0$, & $b_1$, & $b_2$, & $\ldots$, & $b_k$ \\
Agent $3$: & $c_0$  &        &        &           &  
\end{tabular}
\end{center}
for some positive integer $k$, where $\min \{ a_0, b_0, c_0 \} \geq \max \{ \sum_{\ell=1}^k a_\ell, \sum_{\ell=1}^k b_\ell \}$. Then, there exists an EF1 exchange-and-transfer path of length $k + 2$ between this allocation and the allocation with $\{a_1, \ldots, a_k\}$ and $\{b_1, \ldots, b_k\}$ exchanged.
\end{theorem}

In order to prove the result, we shall make use of the following three lemmas in the \emph{exchange-only} context.

\begin{lemma} \label{lem:result_1}
Suppose three agents with identical utilities have goods with the following utilities:
\begin{center}
\begin{tabular}{llllll}
Agent $1$: & $1$, & $a_1$, & $a_2$, & $\ldots$, & $a_k$ \\
Agent $2$: & $1$, & $b_1$, & $b_2$, & $\ldots$, & $b_k$ \\
Agent $3$: & $1$, & $c$,   & $0$,   & $\ldots$, & $0$
\end{tabular}
\end{center}
for some positive integer $k$, where
\begin{itemize}
    \item $c \leq 1$,
    \item $c \geq a_1 \geq \cdots \geq a_k$,
    \item $c \geq b_1 \geq \cdots \geq b_k$,
    \item $\sum_{\ell=1}^k a_\ell \leq 1$, and
    \item $\sum_{\ell=1}^k b_\ell \leq 1$.
\end{itemize}
Then, there exists an EF1 exchange path of length $k$ between this allocation and the allocation with $\{a_1, \ldots, a_k\}$ and $\{b_1, \ldots, b_k\}$ exchanged (i.e., the exchange path is \emph{optimal}).
\end{lemma}

\begin{proof}
We shall refer to a good by its utility. 
Note that the initial and target allocations are EF1. 
It suffices to show that at each step, there exists some $\ell$ such that $a_\ell$ in agent~$1$'s bundle and $b_\ell$ in agent $2$'s bundle can be exchanged while keeping the allocation EF1. 
Then, the EF1 exchange path has length $k$ and is hence optimal.

Suppose that there exists a subset $L$ of $K = \{1, 2, \dots, k\}$ such that $\{a_\ell \; | \; \ell \in L\}$ and $\{b_\ell \; | \; \ell \in L\}$ have already been exchanged from the initial allocation. Then, the current EF1 allocation is
\begin{center}
\begin{tabular}{llll}
Agent $1$: & $1$, & $\{b_\ell \; | \; \ell \in L \}$, & $\{a_j \; | \; j \in K \setminus L\}$ \\
Agent $2$: & $1$, & $\{a_\ell \; | \; \ell \in L \}$, & $\{b_j \; | \; j \in K \setminus L\}$ \\
Agent $3$: & $1$, & $c, 0, \ldots, 0$.                &
\end{tabular}
\end{center}

We claim that there exists some $j \in K \setminus L$ such that exchanging $a_j$ and $b_j$ still maintains EF1. 
If $a_j \geq b_j$ for all $j \in K \setminus L$, then repeatedly exchanging $a_j$ with $b_j$ will cause the utility of agent $1$'s bundle to monotonically decrease and that of agent $2$'s bundle to monotonically increase, and since the target allocation is EF1, all intermediate allocations are also EF1. 
A similar argument holds if $a_j < b_j$ for all $j \in K \setminus L$. 
Therefore, assume that there exists a partition of $K \setminus L$ into non-empty sets $X$ and $Y$ such that $a_x \geq b_x$ for all $x \in X$ and $a_y < b_y$ for all $y \in Y$. 
Let $A$ and $B$ be the utilities of agent~$1$'s bundle and agent $2$'s bundle in the current allocation, respectively, and assume without loss of generality that $A \geq B$. We show that some $j \in X$ or $j \in Y$ works.

If exchanging $a_x$ and $b_x$ leads to an EF1 allocation for some $x \in X$, then we perform the exchange. 
Similarly, if exchanging $a_y$ and $b_y$ leads to an EF1 allocation for some $y \in Y$, then we perform the exchange.
Otherwise, assume that each of these exchanges does not lead to an EF1 allocation. Consider exchanging $a_x$ and $b_x$ for some $x \in X$. 
By assumption, the new allocation is not EF1. 
Note that each of agents $1$ and $2$ has a bundle with utility at least~$1$, so none of them envies agent $3$ by more than one good. 
Since agent $2$'s bundle is the only bundle to increase in value, some agent envies agent $2$ by more than one good. 
Now, $A + B = (1 + \sum_{\ell \in K} a_\ell) + (1 + \sum_{\ell \in K} b_\ell) \leq 4$ and $A \geq B$, which implies $B \leq 2$. 
Agent $2$'s new bundle without the most valuable good would have utility at most $(B - 1) + a_x - b_x \leq 1 + a_x - 0 \leq 1 + c$; note that $1 + c$ is the utility of agent $3$'s bundle. 
Hence, agent $3$ does not envy agent $2$ by more than one good, and so agent $1$ must be the only agent who envies agent $2$ by more than one good. 
This means that, for each $x \in X$, exchanging $a_x$ and $b_x$ leads to agent~$1$ envying agent $2$ by more than one good.

It follows that for every $x \in X$, we have $A - (a_x - b_x) < (B - 1) + (a_x - b_x)$, which implies $2(a_x - b_x) > A - B + 1$. 
Since $A \geq B$ and $b_x \geq 0$, we have $a_x > 0.5$. 
This means that $c \geq a_x > 0.5$, and agent $3$'s bundle has utility at least $1.5$.

Similarly, consider exchanging $a_y$ and $b_y$ for some $y \in Y$. 
Since agent $1$'s bundle is the only bundle to increase in value, some agent envies agent $1$ by more than one good. 
If agent~$3$ envies agent $1$ by more than one good, then agent $3$'s bundle having utility at least $1.5$ implies that agent $1$'s bundle, including the good with utility $1$, now has utility more than $2.5$, i.e., $A + (b_y - a_y) > 2.5$. 
Agent $2$'s bundle now has utility $(A + B) - (A + (b_y - a_y)) < 4 - 2.5 = 1.5$, which is less than the utility of agent $3$'s bundle, so agent $2$ envies agent $1$ by more than one good as well. 
We conclude that exchanging $a_y$ and $b_y$ leads to agent $2$ envying agent $1$ by more than one good for each $y \in Y$.

So far, we have shown that if for every $j\in K\setminus L$ exchanging $a_j$ and $b_j$ does not result in an EF1 allocation, then for each $j\in K\setminus L$, exchanging $a_j$ and $b_j$ will result in agent~$1$ envying agent~$2$ by more than one good or agent $2$ envying agent $1$ by more than one good. 
But this contradicts (the outline\footnote{See the paragraph after the statement of \Cref{thm:binary_2_optimal}.} of) the proof of \Cref{thm:iden_2_optimal} that there always exists $j \in K \setminus L$ such that exchanging $a_j$ and $b_j$ leads to neither agent~$1$ nor agent~$2$ envying each other by more than one good (ignoring agent~$3$).
Therefore, one of these pairs of goods can be exchanged to give an EF1 allocation.
\end{proof}

\begin{lemma} \label{lem:result_2}
Suppose three agents with identical utilities have goods with the following utilities:
\begin{center}
\begin{tabular}{llllll}
Agent $1$: & $1$, & $a_1$, & $a_2$, & $\ldots$, & $a_k$ \\
Agent $2$: & $1$, & $b_1$, & $b_2$, & $\ldots$, & $b_k$ \\
Agent $3$: & $1$, & $0$,   & $0$,   & $\ldots$, & $0$
\end{tabular}
\end{center}
for some positive integer $k$, where $\sum_{\ell=1}^k a_\ell \leq 1$ and $\sum_{\ell=1}^k b_\ell \leq 1$. Then, there exists an EF1 exchange path of length $k + 2$ between this allocation and the allocation with $\{a_1, \ldots, a_k\}$ and $\{b_1, \ldots, b_k\}$ exchanged. 
\end{lemma}

\begin{proof}
Note that the initial and target allocations are EF1.
Without loss of generality, we may assume that $a_1 \geq \cdots \geq a_k$ and $b_1 \geq \cdots \geq b_k$, and moreover that $a_1\ge b_1$. 
We refer to a good by its utility and to a good with utility~$0$ as a dummy good. 
First, exchange $a_1$ with a dummy good in agent~$3$'s bundle. 
Note that the resulting allocation is EF1 and satisfies the assumptions of \Cref{lem:result_1}. 
Hence, by \Cref{lem:result_1}, we can exchange $\{a_2, \ldots, a_k,0\}$ and $\{b_1, \ldots, b_k\}$ in $k$ steps while maintaining EF1. 
Finally, exchange $a_1$ with agent~$3$'s original dummy good in agent~$2$'s bundle. 
We have arrived at the target allocation, and the total number of exchanges is $k+2$.
\end{proof}

\begin{lemma} \label{lem:result_3}
Suppose three agents with identical utilities have goods with the following utilities:
\begin{center}
\begin{tabular}{llllll}
Agent $1$: & $a_0$, & $a_1$, & $a_2$, & $\ldots$, & $a_k$ \\
Agent $2$: & $b_0$, & $b_1$, & $b_2$, & $\ldots$, & $b_k$ \\
Agent $3$: & $c_0$, & $0$,   & $0$,   & $\ldots$, & $0$
\end{tabular}
\end{center}
for some positive integer $k$, where $\min\{a_0, b_0, c_0\} \geq \max \{ \sum_{\ell=1}^k a_\ell, \sum_{\ell=1}^k b_\ell \}$. Then, there exists an EF1 exchange path of length $k + 2$ between this allocation and the allocation with $\{a_1, \ldots, a_k\}$ and $\{b_1, \ldots, b_k\}$ exchanged.
\end{lemma}

\begin{proof}
We refer to a good by its utility.
Note that the initial and target allocations are EF1. 
We shall describe an exchange path such that $a_0$, $b_0$, and $c_0$ are not moved. 
Note that $a_0$, $b_0$, and $c_0$ are the three most valuable goods among all the goods, and without loss of generality, assume that they all have positive value. 
Let $m_0 = \min \{a_0, b_0, c_0\} > 0$. 
First, replace $a_0$, $b_0$, and $c_0$ by $m_0$ each. 
For an EF1 allocation where these goods remain with the respective agents (after the replacement), the same allocation but with $m_0$ replaced back by $a_0$, $b_0$, and $c_0$ respectively is also EF1. 
Next, scale all the utilities by a factor of $1/m_0$---note that scaling by a positive factor does not affect the EF1 property of the allocations. 
Now, the initial allocation follows the setting in \Cref{lem:result_2}. 
Hence, by \Cref{lem:result_2}, we can exchange $\{a_1, \ldots, a_k\}$ and $\{b_1, \ldots, b_k\}$ in $k + 2$ steps while maintaining EF1. 
Finally, scale the utilities by a factor of $m_0$, and replace $m_0$ with $a_0$, $b_0$, and $c_0$ respectively to get the desired exchange path.
\end{proof}

We are now ready to prove \Cref{thm:exchange_and_transfer}.

\begin{proof}[Proof of \Cref{thm:exchange_and_transfer}]
Consider the scenario where agent~$3$ is given $k$ additional goods with zero utility. 
This is exactly the scenario in \Cref{lem:result_3}, which means that there is an EF1 exchange path of length $k+2$ between this allocation and the allocation with $\{a_1, \ldots, a_k\}$ and $\{b_1, \ldots, b_k\}$ exchanged. 
Consider this EF1 exchange path. 
Whenever an exchange is performed which involves one of the additional zero-utility goods in agent $3$'s bundle, it can be thought of as a \emph{transfer} of a good in the original scenario where agent $3$ did not receive the zero-utility goods. 
This means that there is an EF1 exchange-and-transfer path of length $k+2$.
\end{proof}

Finally, we remark that the counterexample in \Cref{thm:binary_3_connected} for binary utilities is likewise no longer applicable in the exchange-and-transfer setting.
On the other hand, the example in \Cref{thm:gen_2_connected} for two agents with general utilities \emph{is} still a counterexample when viewed in this setting.

\section{Proof of \texorpdfstring{\Cref{lem:directed_triangle_partition_nphard}}{Lemma 4.2}}
\label{app:directedtrianglepartition}

We prove \Cref{lem:directed_triangle_partition_nphard} by reducing the well-known NP-hard problem \textsc{3SAT} to \textsc{Directed Triangle Partition}.
\begin{itemize}
    \item \textbf{\textsc{3SAT}. } Given a set of variables $Y = \{y_1, \ldots, y_q\}$ and a set of clauses $C = \{c_1, \ldots, c_r\}$ where each clause is a disjunction of three literals, i.e., $c_j = \ell_{j, 1} \lor \ell_{j, 2} \lor \ell_{j, 3}$, and each literal is either a variable (i.e., $\ell_{j, k} = y_i$) or its negation (i.e., $\ell_{j, k} = \overline{y_i}$), determine whether there exists an assignment to the variables in $Y$ such that every clause in $C$ is satisfied.
    \item \textbf{\textsc{Directed Triangle Partition}. } Given a directed graph $G = (V, E)$ with no directed cycles of length $1$ or $2$, determine whether there is a partition of the edges into triangles (i.e., directed cycles of length $3$).
\end{itemize}

The reduction works by constructing a graph for each variable and each literal that can be edge-partitioned into triangles in exactly two ways---one representing ``true'' and the other representing ``false''---and joining these graphs together in special ways to restrict the truth values that they represent. 
This idea is similar to that used by \citet{Holyer81} in his proof of the corresponding result for \emph{undirected} graphs.

\begin{figure*}[tb]
\centering
  \begin{subfigure}[b]{.60\textwidth}
    \centering 
    \begin{tikzpicture}
    \usetikzlibrary{positioning}
    \newdimen\nodeDist
    \nodeDist=12mm
    \newcommand{\arrow}{ \tikz \draw[-{Stealth[length=1.5mm, width=1.5mm]}] (-1pt,0) -- (1pt,0); }
    \node[label=below:{\footnotesize $(0, 0, 0)$}] at (0,0) (v00) {};
    \draw [black, fill=black] (v00) circle [radius=0.05];
    \foreach \x\y\a\b in {0/1/1/3, 1/2/2/2, 2/3/3/1, 3/4/0/0} {
      \node at ([shift=({60:\nodeDist})]v0\x) (v0\y) {};
      \draw [black, fill=black] (v0\y) circle [radius=0.05];
      \node at ([shift=({0:\nodeDist})]v\x0) (v\y0) {};
      \draw [black, fill=black] (v\y0) circle [radius=0.05];
    }
    \foreach \y in {1, ..., 4} {
      \foreach \x\z in {0/1, 1/2, 2/3, 3/4} {
        \node at ([shift=({0:\nodeDist})]v\x\y) (v\z\y) {};
        \draw [black, fill=black] (v\z\y) circle [radius=0.05];
      }
    }
    \foreach \y\z in {1/3, 2/2, 3/1} {
      \node[label=left:{\footnotesize $(\y, 0, \z)$}] at (v0\y) (v0\y p) {};
      \node[label=below:{\footnotesize $(\y, \z, 0)$}] at (v\y0) (v\y0p) {};
      \node[label=right:{\footnotesize $(\y, 0, \z)$}] at (v4\y) (v4\y p) {};
      \node[label=above:{\footnotesize $(\y, \z, 0)$}] at (v\y4) (v\y4p) {};
    }
    \node[label=above:{\footnotesize $(0, 0, 0)$}] at (v04) (v04p) {};
    \node[label=above:{\footnotesize $(0, 0, 0)$}] at (v44) (v44p) {};
    \node[label=below:{\footnotesize $(0, 0, 0)$}] at (v40) (v40p) {};
    \foreach \w\x in {0/1, 1/2, 2/3, 3/4} {
      \foreach \y\z in  {0/1, 1/2, 2/3, 3/4} {
        \draw (v\w\z.center) -- (v\w\y.center) node[
          sloped, midway, allow upside down
        ]{\arrow};;
        \draw (v\y\w.center) -- (v\z\w.center) node[
          sloped, midway, allow upside down
        ]{\arrow};;
        \draw (v\x\y.center) -- (v\w\z.center) node[
          sloped, midway, allow upside down
        ]{\arrow};;
      }
    }
    \foreach \y\z in  {0/1, 1/2, 2/3, 3/4} {
      \draw[dashed] (v4\z.center) -- (v4\y.center);
      \draw[dashed] (v\y4.center) -- (v\z4.center);
    }
    \node at ([shift=({180:1.2*\nodeDist})]v00) (l0) {};
    \node at ([shift=({180:1.2*\nodeDist})]v04) (l1) {};
    \draw [line width=0.4mm, -{Stealth}{Stealth}] (l0.center) -- (l1.center);
    \node at ([shift=({0:1.2*\nodeDist})]v40) (r0) {};
    \node at ([shift=({0:1.2*\nodeDist})]v44) (r1) {};
    \draw [line width=0.4mm, -{Stealth}{Stealth}] (r0.center) -- (r1.center);
    \node at ([shift=({-90:0.6*\nodeDist})]v00) (d0) {};
    \node at ([shift=({-90:0.6*\nodeDist})]v40) (d1) {};
    \draw [line width=0.4mm, -Stealth] (d0.center) -- (d1.center);
    \node at ([shift=({90:0.6*\nodeDist})]v04) (u0) {};
    \node at ([shift=({90:0.6*\nodeDist})]v44) (u1) {};
    \draw [line width=0.4mm, -Stealth] (u0.center) -- (u1.center);
    \end{tikzpicture}
    \caption{The graph $H_4$. The opposite sides wrap around.}
  \end{subfigure}
\begin{subfigure}[b]{.35\textwidth}
    \centering 
    \begin{tikzpicture}
    \usetikzlibrary{positioning}
    \newdimen\nodeDist
    \nodeDist=12mm
    \newcommand{\arrow}{ \tikz \draw[-{Stealth[length=1.5mm, width=1.5mm]}] (-1pt,0) -- (1pt,0); }
    \node at (0,0) (t0) {};
    \node at ([shift=({0:\nodeDist})]t0) (t1) {};
    \node at ([shift=({120:\nodeDist})]t1) (t2) {};
    \node at ([shift=({-90:1.5*\nodeDist})]t0) (f0) {};
    \node at ([shift=({0:\nodeDist})]f0) (f1) {};
    \node at ([shift=({-120:\nodeDist})]f1) (f2) {};
    \foreach \x\y in {0/1, 1/2, 2/0} {
      \draw [black, fill=black] (t\x) circle [radius=0.05];
      \draw [black, fill=black] (f\x) circle [radius=0.05];
      \draw (t\x.center) -- (t\y.center) node[
        sloped, midway, allow upside down
      ]{\arrow};;
      \draw (f\x.center) -- (f\y.center) node[
        sloped, midway, allow upside down
      ]{\arrow};;
    }
    \node at ([shift=({-90:1.5*\nodeDist})]f0) (v) {}; 
    \end{tikzpicture}
    \caption{A $T$-triangle (top) and an $F$-triangle (bottom).}
  \end{subfigure}
\caption{The graph $H_4$ and an example of a $T$-triangle and an $F$-triangle.} \label{fig:mep_Hp}
\end{figure*}
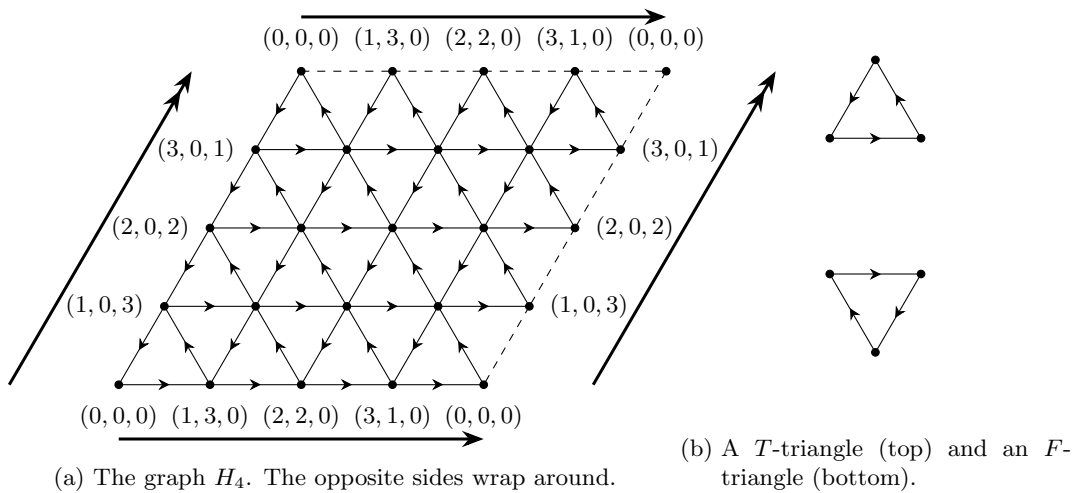

Define the directed graph $H_p = (V_p, E_p)$ for each positive integer $p$ as follows:
\begin{align*}
    V_p &= \{ (a_1, a_2, a_3) \in \mathbb{Z}_p^3 \mid a_1 + a_2 + a_3 \equiv 0 \}, \\
    E_p &= \{ (a_1, a_2, a_3) \to (b_1, b_2, b_3) \mid 
    \exists (i, j, k) \in \{(1, 2, 3), (2, 3, 1), (3, 1, 2)\} \text{ such that } \\
    & \hspace{50mm} b_i \equiv a_i, \; b_j \equiv a_j + 1, \text{ and } b_k \equiv a_k - 1 \},
\end{align*}
where all equivalences are modulo $p$. 
There are only two types of triangles in $H_p$: \emph{$T$-triangles} of the form $(a_1, a_2, a_3) \to (a_1 + 1, a_2 - 1, a_3) \to (a_1 + 1, a_2, a_3 - 1) \to (a_1, a_2, a_3)$, and \emph{$F$-triangles} of the form $(a_1, a_2, a_3) \to (a_1 + 1, a_2 - 1, a_3) \to (a_1, a_2 - 1, a_3 + 1) \to (a_1, a_2, a_3)$. 
See \Cref{fig:mep_Hp} for an illustration.
Note that each vertex has an indegree and an outdegree of $3$.

\begin{figure}[tb]
\centering
  \begin{subfigure}[b]{.23\textwidth}
    \centering 
    \begin{tikzpicture}
    \usetikzlibrary{positioning}
    \newdimen\nodeDist
    \nodeDist=8.2mm
    \newcommand{\arrow}{ \tikz \draw[-{Stealth[length=1.5mm, width=1.5mm]}] (-1pt,0) -- (1pt,0); }
    \node at (0,0) (v00) {};
    \foreach \x\y in {0/1, 1/2, 2/3, 3/4} {
      \node at ([shift=({-120:\nodeDist})]v0\x) (v0\y) {};
      \node at ([shift=({180:\nodeDist})]v\x0) (v\y0) {};
    }
    \node at ([shift=({180:\nodeDist})]v40) (v50) {};
    \foreach \y in {1, ..., 4} {
      \foreach \x\z in {0/1, 1/2, 2/3, 3/4} {
        \pgfmathparse{\x+\y <= 4 ? int(1) : int(0)}
        \ifnum\pgfmathresult=1
          \node at ([shift=({180:\nodeDist})]v\x\y) (v\z\y) {};
        \fi
      }
    }
    \draw [black, fill=black] (v11) circle [radius=0.05];
    \draw [black, fill=black] (v21) circle [radius=0.05];
    \draw [black, fill=black] (v31) circle [radius=0.05];
    \draw [black, fill=black] (v12) circle [radius=0.05];
    \draw [black, fill=black] (v22) circle [radius=0.05];
    \draw [black, fill=black] (v13) circle [radius=0.05];
    \draw (v21.center) -- (v11.center) node[sloped, midway, allow upside down]{\arrow};;
    \draw (v31.center) -- (v21.center) node[sloped, midway, allow upside down]{\arrow};;
    \draw (v22.center) -- (v31.center) node[sloped, midway, allow upside down]{\arrow};;
    \draw (v13.center) -- (v22.center) node[sloped, midway, allow upside down]{\arrow};;
    \draw (v12.center) -- (v13.center) node[sloped, midway, allow upside down]{\arrow};;
    \draw (v11.center) -- (v12.center) node[sloped, midway, allow upside down]{\arrow};;
    \draw [line width=0.6mm, blue] (v12.center) -- (v21.center) node[sloped, midway, allow upside down]{\arrow};;
    \draw [line width=0.6mm, blue] (v22.center) -- (v12.center) node[sloped, midway, allow upside down]{\arrow};;
    \draw [line width=0.6mm, blue] (v21.center) -- (v22.center) node[sloped, midway, allow upside down]{\arrow};;
    \foreach \w\x in {1/2, 2/3, 3/4} {
      \draw[dotted] (v0\w.center) -- (v1\w.center);
      \draw[dotted] (v\w0.center) -- (v\w1.center);
      \draw[dotted] (v0\x.center) -- (v1\w.center);
      \draw[dotted] (v\x0.center) -- (v\w1.center);
    }
    \draw[dotted] (v13.center) -- (v14.center);
    \draw[dotted] (v22.center) -- (v23.center);
    \draw[dotted] (v31.center) -- (v32.center);
    \draw[dotted] (v31.center) -- (v41.center);
    \draw[dotted] (v22.center) -- (v32.center);
    \draw[dotted] (v13.center) -- (v23.center);
    \end{tikzpicture}
    \caption{A $T$-patch.}
  \end{subfigure}
  \begin{subfigure}[b]{.23\textwidth}
    \centering 
    \begin{tikzpicture}
    \usetikzlibrary{positioning}
    \newdimen\nodeDist
    \nodeDist=8.2mm
    \newcommand{\arrow}{ \tikz \draw[-{Stealth[length=1.5mm, width=1.5mm]}] (-1pt,0) -- (1pt,0); }
    \node at (0,0) (v00) {};
    \foreach \x\y in {0/1, 1/2, 2/3, 3/4} {
      \node at ([shift=({60:\nodeDist})]v0\x) (v0\y) {};
      \node at ([shift=({0:\nodeDist})]v\x0) (v\y0) {};
    }
    \node at ([shift=({0:\nodeDist})]v40) (v50) {};
    \foreach \y in {1, ..., 4} {
      \foreach \x\z in {0/1, 1/2, 2/3, 3/4} {
        \pgfmathparse{\x+\y <= 4 ? int(1) : int(0)}
        \ifnum\pgfmathresult=1
          \node at ([shift=({0:\nodeDist})]v\x\y) (v\z\y) {};
        \fi
      }
    }
    \draw [black, fill=black] (v11) circle [radius=0.05];
    \draw [black, fill=black] (v21) circle [radius=0.05];
    \draw [black, fill=black] (v31) circle [radius=0.05];
    \draw [black, fill=black] (v12) circle [radius=0.05];
    \draw [black, fill=black] (v22) circle [radius=0.05];
    \draw [black, fill=black] (v13) circle [radius=0.05];
    \draw (v11.center) -- (v21.center) node[sloped, midway, allow upside down]{\arrow};;
    \draw (v21.center) -- (v31.center) node[sloped, midway, allow upside down]{\arrow};;
    \draw (v31.center) -- (v22.center) node[sloped, midway, allow upside down]{\arrow};;
    \draw (v22.center) -- (v13.center) node[sloped, midway, allow upside down]{\arrow};;
    \draw (v13.center) -- (v12.center) node[sloped, midway, allow upside down]{\arrow};;
    \draw (v12.center) -- (v11.center) node[sloped, midway, allow upside down]{\arrow};;
    \draw [line width=0.6mm, blue] (v21.center) -- (v12.center) node[sloped, midway, allow upside down]{\arrow};;
    \draw [line width=0.6mm, blue] (v12.center) -- (v22.center) node[sloped, midway, allow upside down]{\arrow};;
    \draw [line width=0.6mm, blue] (v22.center) -- (v21.center) node[sloped, midway, allow upside down]{\arrow};;
    \foreach \w\x in {1/2, 2/3, 3/4} {
      \draw[dotted] (v0\w.center) -- (v1\w.center);
      \draw[dotted] (v\w0.center) -- (v\w1.center);
      \draw[dotted] (v0\x.center) -- (v1\w.center);
      \draw[dotted] (v\x0.center) -- (v\w1.center);
    }
    \draw[dotted] (v13.center) -- (v14.center);
    \draw[dotted] (v22.center) -- (v23.center);
    \draw[dotted] (v31.center) -- (v32.center);
    \draw[dotted] (v31.center) -- (v41.center);
    \draw[dotted] (v22.center) -- (v32.center);
    \draw[dotted] (v13.center) -- (v23.center);
    \end{tikzpicture}
    \caption{An $F$-patch.}
  \end{subfigure}
\caption{A $T$-patch and an $F$-patch. The center $T$-triangle and $F$-triangle are denoted by bold lines, and the exteriors are denoted by non-bold solid lines. The dotted lines are not part of the patches.} \label{fig:mep_patch}
\end{figure}
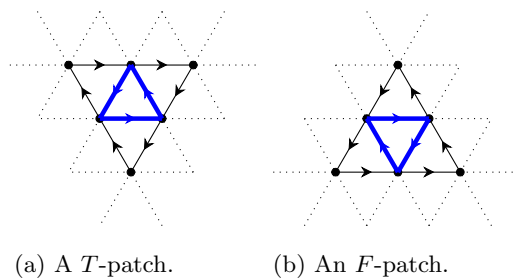

Two triangles are called \emph{neighbors} if they share a common edge. 
A \emph{patch} is a triangle together with its neighbors. The set of edges of the triangle is called the \emph{center} of the patch, and the set of edges of a patch that do not belong to the center is called the \emph{exterior} of the patch. A \emph{$T$-patch} (resp., \emph{$F$-patch}) is a patch in which the center is a $T$-triangle (resp., $F$-triangle). 
See \Cref{fig:mep_patch} for an illustration. 
Two patches $P^1$ and $P^2$ are \emph{non-interfering} if the distance between any vertex in $P^1$ and any vertex in $P^2$ is at least (say) $10$ on $H_p$, where distance is measured along a shortest path. 
We shall also require patches to be of distance at least $10$ from the vertex $\mathbf{0} = (0, 0, 0)$.

Consider the graph $H_p$ with non-interfering patches and with some edges of the patches removed. 
Suppose there is an edge-partition of the resulting graph into triangles. 
The vertex $\mathbf{0}$ has an indegree and an outdegree of $3$, so any edge-partition into triangles requires $\mathbf{0}$ to belong to exactly \emph{three} triangles. 
The only ways to have these three triangles are when they are all $T$-triangles or all $F$-triangles. 
Then, all neighboring vertices to $\mathbf{0}$ belong to triangles of the same type. 
By a similar argument, the neighboring vertices must each belong to exactly three triangles of the same type. 
This cascades through the whole graph (except possibly at the patches), and therefore, we see that an edge-partition of $H_p$ into triangles necessarily consists only of $T$-triangles or only of $F$-triangles (except possibly at the patches).

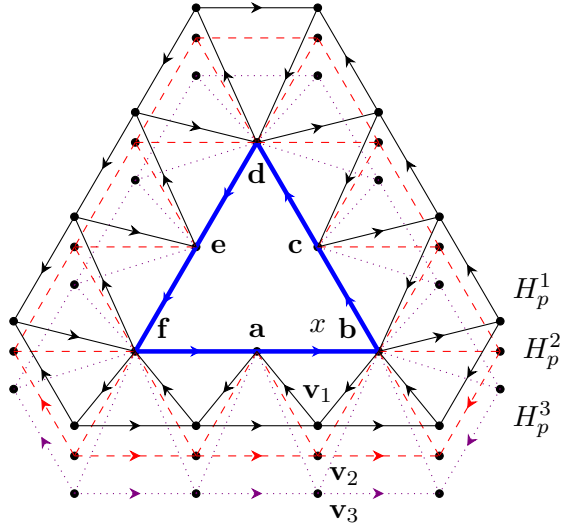
\begin{figure}[htb]
\centering
  \begin{subfigure}[b]{.48\textwidth}
    \centering 
    \begin{tikzpicture}
    \usetikzlibrary{positioning}
    \newdimen\nodeDist
    \nodeDist=16mm
    \newcommand{\arrow}{ \tikz \draw[-{Stealth[length=1.5mm, width=1.5mm]}] (-1pt,0) -- (1pt,0); }
    \node at (0,-5mm) (v100) {};
    \node at (0,0) (v200) {};
    \node at (0,4mm) (v300) {};
    \foreach \h in {1, 2, 3} {
      \foreach \x\y in {0/1, 1/2, 2/3, 3/4} {
        \node at ([shift=({60:\nodeDist})]v\h0\x) (v\h0\y) {}; 
        \draw [black, fill=black] (v\h0\y) circle [radius=0.05];
        \node at ([shift=({0:\nodeDist})]v\h\x0) (v\h\y0) {}; 
        \draw [black, fill=black] (v\h\y0) circle [radius=0.05];
      }
      \node at ([shift=({0:\nodeDist})]v\h04) (v\h05) {}; 
      \node at ([shift=({0:\nodeDist})]v\h40) (v\h50) {}; 
      \foreach \y in {1, ..., 4} {
        \foreach \x\z in {0/1, 1/2, 2/3, 3/4} {
          \pgfmathparse{\x+\y <= 4 ? int(1) : int(0)}
          \ifnum\pgfmathresult=1
            \node at ([shift=({0:\nodeDist})]v\h\x\y) (v\h\z\y) {};
          \fi
        }
      }
      \draw [black, fill=black] (v\h14) circle [radius=0.05];
      \draw [black, fill=black] (v\h23) circle [radius=0.05];
      \draw [black, fill=black] (v\h32) circle [radius=0.05];
      \draw [black, fill=black] (v\h41) circle [radius=0.05];
    } 
    \draw [black, fill=black] (v211) circle [radius=0.05] node[above right=0.4mm and 1.5mm] {$\mathbf{f}$};
    \draw [black, fill=black] (v221) circle [radius=0.05] node[above=0.6mm] {$\mathbf{a}$};
    \draw [black, fill=black] (v231) circle [radius=0.05] node[above left=0.4mm and 1.5mm] {$\mathbf{b}$};
    \draw [black, fill=black] (v222) circle [radius=0.05] node[left=0.6mm] {$\mathbf{c}$};
    \draw [black, fill=black] (v213) circle [radius=0.05] node[below=1.4mm] {$\mathbf{d}$};
    \draw [black, fill=black] (v212) circle [radius=0.05] node[right=0.6mm] {$\mathbf{e}$};
    \draw [line width=0.6mm, blue] (v211.center) -- (v221.center) node[sloped, midway, allow upside down]{\arrow};;
    \draw [line width=0.6mm, blue] (v221.center) -- node[above=1mm, black] {$x$} (v231.center) node[sloped, midway, allow upside down]{\arrow};;
    \draw [line width=0.6mm, blue] (v231.center) -- (v222.center) node[sloped, midway, allow upside down]{\arrow};;
    \draw [line width=0.6mm, blue] (v222.center) -- (v213.center) node[sloped, midway, allow upside down]{\arrow};;
    \draw [line width=0.6mm, blue] (v213.center) -- (v212.center) node[sloped, midway, allow upside down]{\arrow};;
    \draw [line width=0.6mm, blue] (v212.center) -- (v211.center) node[sloped, midway, allow upside down]{\arrow};;
    \foreach \x\y\w\z in {1/2/3/4, 2/3/2/3, 3/4/1/2} {
      \draw [dotted, violet] (v1\x0.center) -- (v1\y0.center) node[sloped, midway, allow upside down]{\arrow};;
      \draw [dotted, violet] (v2\x1.center) -- (v1\x0.center);
      \draw [dotted, violet] (v1\y0.center) -- (v2\x1.center);
      \draw [dotted, violet] (v10\y.center) -- (v10\x.center);
      \draw [dotted, violet] (v10\x.center) -- (v21\x.center);
      \draw [dotted, violet] (v21\x.center) -- (v10\y.center);
      \draw [dotted, violet] (v2\x\w.center) -- (v1\y\w.center);
      \draw [dotted, violet] (v1\y\w.center) -- (v1\x\z.center);
      \draw [dotted, violet] (v1\x\z.center) -- (v2\x\w.center);
    }
    \draw [dotted, violet] (v110.center) -- (v101.center) node[sloped, midway, allow upside down]{\arrow};;
    \draw [dotted, violet] (v141.center) -- (v140.center) node[sloped, midway, allow upside down]{\arrow};;
    \draw [dotted, violet] (v104.center) -- (v114.center);
    \foreach \x\y\w\z in {1/2/3/4, 2/3/2/3, 3/4/1/2} {
      \draw [dashed, red] (v2\x0.center) -- (v2\y0.center) node[sloped, midway, allow upside down]{\arrow};;
      \draw [dashed, red] (v2\x1.center) -- (v2\x0.center);
      \draw [dashed, red] (v2\y0.center) -- (v2\x1.center);
      \draw [dashed, red] (v20\y.center) -- (v20\x.center);
      \draw [dashed, red] (v20\x.center) -- (v21\x.center);
      \draw [dashed, red] (v21\x.center) -- (v20\y.center);
      \draw [dashed, red] (v2\x\w.center) -- (v2\y\w.center);
      \draw [dashed, red] (v2\y\w.center) -- (v2\x\z.center);
      \draw [dashed, red] (v2\x\z.center) -- (v2\x\w.center);
    }
    \draw [dashed, red] (v210.center) -- (v201.center) node[sloped, midway, allow upside down]{\arrow};;
    \draw [dashed, red] (v241.center) -- (v240.center) node[sloped, midway, allow upside down]{\arrow};;
    \draw [dashed, red] (v204.center) -- (v214.center);
    \foreach \x\y\w\z in {1/2/3/4, 2/3/2/3, 3/4/1/2} {
      \draw [] (v3\x0.center) -- (v3\y0.center) node[sloped, midway, allow upside down]{\arrow};;
      \draw [] (v2\x1.center) -- (v3\x0.center) node[sloped, midway, allow upside down]{\arrow};;
      \draw [] (v3\y0.center) -- (v2\x1.center) node[sloped, midway, allow upside down]{\arrow};;
      \draw [] (v30\y.center) -- (v30\x.center) node[sloped, midway, allow upside down]{\arrow};;
      \draw [] (v30\x.center) -- (v21\x.center) node[sloped, midway, allow upside down]{\arrow};;
      \draw [] (v21\x.center) -- (v30\y.center) node[sloped, midway, allow upside down]{\arrow};;
      \draw [] (v2\x\w.center) -- (v3\y\w.center) node[sloped, midway, allow upside down]{\arrow};;
      \draw [] (v3\y\w.center) -- (v3\x\z.center) node[sloped, midway, allow upside down]{\arrow};;
      \draw [] (v3\x\z.center) -- (v2\x\w.center) node[sloped, midway, allow upside down]{\arrow};;
    }
    \draw [] (v310.center) -- (v301.center) node[sloped, midway, allow upside down]{\arrow};;
    \draw [] (v341.center) -- (v340.center) node[sloped, midway, allow upside down]{\arrow};;
    \draw [] (v304.center) -- (v314.center) node[sloped, midway, allow upside down]{\arrow};;
    \draw (v141) node[below right=0.4mm] {$H_p^3$};
    \draw (v241) node[right=1.6mm] {$H_p^2$};
    \draw (v341) node[above right=0.4mm] {$H_p^1$};
    \draw (v130) node[below right=0.2mm] {$\mathbf{v}_3$};
    \draw (v230) node[below right=0.2mm] {$\mathbf{v}_2$};
    \draw (v330) node[above=1.8mm] {$\mathbf{v}_1$};
    \end{tikzpicture}
  \end{subfigure}
\caption{An $F$-$F$-$F$ join on $H_p^1$ (solid lines), $H_p^2$ (dashed lines) and $H_p^3$ (dotted lines). All three graphs share the exterior of the patch $P_F$ (bold lines).} \label{fig:mep_FFF}
\end{figure}

Let $H_p^1$, $H_p^2$, and $H_p^3$ be three copies of $H_p$, and let $P_F^k$ be an $F$-patch on $H_p^k$ for each $k \in \{1, 2, 3\}$. 
We say that we apply an \emph{$F$-$F$-$F$ join} on $(H_p^1, H_p^2, H_p^3)$ if we remove the patches $P_F^1$, $P_F^2$, and $P_F^3$ on the respective copies and replace them by \emph{one} copy of the vertices of an $F$-patch $P_F$ and \emph{one} copy of the exterior of $P_F$. See \Cref{fig:mep_FFF} for an illustration. 
We claim that any edge-partition of this new graph into triangles results in \emph{exactly one} of $H_p^1$, $H_p^2$, and $H_p^3$ being partitioned into $F$-triangles (and the other two into $T$-triangles). 
To see this, consider an edge $x = \mathbf{a} \to \mathbf{b}$ belonging to the exterior of $P_F$. Since $x$ belongs to a triangle, we consider the candidates for the third vertex of the triangle. 
There are only three such candidates: $\mathbf{v}_1$, $\mathbf{v}_2$, and $\mathbf{v}_3$, which are parallel vertices on $H_p^1$, $H_p^2$, and $H_p^3$, respectively.

Assume without loss of generality that the triangle is $\mathbf{a} \to \mathbf{b} \to \mathbf{v}_1 \to \mathbf{a}$ (note that this is an $F$-triangle), and consider the triangles containing the vertex $\mathbf{v}_1$. 
Since $\mathbf{v}_1$ already has an $F$-triangle, the other triangles containing $\mathbf{v}_1$ can only be $F$-triangles, and the cascading effect implies that $\mathbf{0}$ in $H_p^1$, and all other vertices in $H_p^1$ (except possibly at~$P_F$), belong to $F$-triangles. 
Note that this implies that each edge in the exterior of~$P_F$ is combined with edges in $H_p^1$ to form $F$-triangles. 
On the other hand, the indegree and the outdegree of both $\mathbf{v}_2$ and $\mathbf{v}_3$ are $3$, so there must be exactly three triangles containing $\mathbf{v}_2$ and $\mathbf{v}_3$, respectively. 
Since $x$ is used by the edges $\mathbf{b} \to \mathbf{v}_1$ and $\mathbf{v}_1 \to \mathbf{a}$ in $H_p^1$, it cannot be used by the edges $\mathbf{b} \to \mathbf{v}_2$ and $\mathbf{v}_2 \to \mathbf{a}$ in $H_p^2$ or $\mathbf{b} \to \mathbf{v}_3$ and $\mathbf{v}_3 \to \mathbf{a}$ in $H_p^3$ to form the respective $F$-triangles.
As such, the triangles containing $\mathbf{v}_2$ and $\mathbf{v}_3$ must all be $T$-triangles, and the cascading effect implies that $\mathbf{0}$ in $H_p^2$ and $H_p^3$, and all other vertices in $H_p^2$ and $H_p^3$ (except at $P_F$), belong to $T$-triangles. 
It can be verified that these edge-partitions into triangles are indeed valid.

\begin{figure}[htb]
\centering
  \begin{subfigure}[b]{.5\textwidth}
    \centering 
    \begin{tikzpicture}
    \usetikzlibrary{positioning}
    \newdimen\nodeDist
    \nodeDist=15mm
    \newcommand{\arrow}{ \tikz \draw[-{Stealth[length=1.5mm, width=1.5mm]}] (-1pt,0) -- (1pt,0); }
    \node at (0,0) (v200) {};
    \node at (0,3mm) (v300) {};
    \foreach \h in {2, 3} {
      \foreach \x\y in {0/1, 1/2, 2/3, 3/4} {
        \node at ([shift=({60:\nodeDist})]v\h0\x) (v\h0\y) {}; 
        \draw [black, fill=black] (v\h0\y) circle [radius=0.05];
        \node at ([shift=({0:\nodeDist})]v\h\x0) (v\h\y0) {}; 
        \draw [black, fill=black] (v\h\y0) circle [radius=0.05];
      }
      \node at ([shift=({0:\nodeDist})]v\h04) (v\h05) {}; 
      \node at ([shift=({0:\nodeDist})]v\h40) (v\h50) {}; 
      \foreach \y in {1, ..., 4} {
        \foreach \x\z in {0/1, 1/2, 2/3, 3/4} {
          \pgfmathparse{\x+\y <= 4 ? int(1) : int(0)}
          \ifnum\pgfmathresult=1
            \node at ([shift=({0:\nodeDist})]v\h\x\y) (v\h\z\y) {};
          \fi
        }
      }
      \draw [black, fill=black] (v\h14) circle [radius=0.05];
      \draw [black, fill=black] (v\h23) circle [radius=0.05];
      \draw [black, fill=black] (v\h32) circle [radius=0.05];
      \draw [black, fill=black] (v\h41) circle [radius=0.05];
    } 
    \draw [black, fill=black] (v211) circle [radius=0.05] node[above=1.6mm] {$\mathbf{f}$};
    \draw [black, fill=black] (v221) circle [radius=0.05] node[above=1.6mm] {$\mathbf{a}$};
    \draw [black, fill=black] (v231) circle [radius=0.05] node[above=1.6mm] {$\mathbf{b}$};
    \draw [black, fill=black] (v222) circle [radius=0.05] node[below=1.6mm] {$\mathbf{c}$};
    \draw [black, fill=black] (v213) circle [radius=0.05] node[below=1.4mm] {$\mathbf{d}$};
    \draw [black, fill=black] (v212) circle [radius=0.05] node[below=1.6mm] {$\mathbf{e}$};
    \draw [line width=0.6mm, blue] (v211.center) -- (v221.center) node[sloped, midway, allow upside down]{\arrow};;
    \draw [line width=0.6mm, blue] (v221.center) --  node[above=0.5mm, black] {$x$} (v231.center) node[sloped, midway, allow upside down]{\arrow};;
    \draw [line width=0.6mm, blue] (v231.center) -- (v222.center) node[sloped, midway, allow upside down]{\arrow};;
    \draw [line width=0.6mm, blue] (v222.center) -- (v213.center) node[sloped, midway, allow upside down]{\arrow};;
    \draw [line width=0.6mm, blue] (v213.center) -- (v212.center) node[sloped, midway, allow upside down]{\arrow};;
    \draw [line width=0.6mm, blue] (v212.center) -- (v211.center) node[sloped, midway, allow upside down]{\arrow};;
    \draw [line width=0.6mm, blue] (v221.center) -- (v212.center) node[sloped, midway, allow upside down]{\arrow};;
    \draw [line width=0.6mm, blue] (v212.center) -- (v222.center) node[sloped, midway, allow upside down]{\arrow};;
    \draw [line width=0.6mm, blue] (v222.center) -- (v221.center) node[sloped, midway, allow upside down]{\arrow};;
    \foreach \x\y\w\z in {1/2/3/4, 2/3/2/3, 3/4/1/2} {
      \draw [dashed, red] (v2\x0.center) -- (v2\y0.center) node[sloped, midway, allow upside down]{\arrow};;
      \draw [dashed, red] (v2\x1.center) -- (v2\x0.center);
      \draw [dashed, red] (v2\y0.center) -- (v2\x1.center);
      \draw [dashed, red] (v20\y.center) -- (v20\x.center);
      \draw [dashed, red] (v20\x.center) -- (v21\x.center);
      \draw [dashed, red] (v21\x.center) -- (v20\y.center);
      \draw [dashed, red] (v2\x\w.center) -- (v2\y\w.center);
      \draw [dashed, red] (v2\y\w.center) -- (v2\x\z.center);
      \draw [dashed, red] (v2\x\z.center) -- (v2\x\w.center);
    }
    \draw [dashed, red] (v210.center) -- (v201.center) node[sloped, midway, allow upside down]{\arrow};;
    \draw [dashed, red] (v241.center) -- (v240.center) node[sloped, midway, allow upside down]{\arrow};;
    \draw [dashed, red] (v204.center) -- (v214.center);
    \foreach \x\y\w\z in {1/2/3/4, 2/3/2/3, 3/4/1/2} {
      \draw [] (v3\x0.center) -- (v3\y0.center) node[sloped, midway, allow upside down]{\arrow};;
      \draw [] (v2\x1.center) -- (v3\x0.center) node[sloped, midway, allow upside down]{\arrow};;
      \draw [] (v3\y0.center) -- (v2\x1.center) node[sloped, midway, allow upside down]{\arrow};;
      \draw [] (v30\y.center) -- (v30\x.center) node[sloped, midway, allow upside down]{\arrow};;
      \draw [] (v30\x.center) -- (v21\x.center) node[sloped, midway, allow upside down]{\arrow};;
      \draw [] (v21\x.center) -- (v30\y.center) node[sloped, midway, allow upside down]{\arrow};;
      \draw [] (v2\x\w.center) -- (v3\y\w.center) node[sloped, midway, allow upside down]{\arrow};;
      \draw [] (v3\y\w.center) -- (v3\x\z.center) node[sloped, midway, allow upside down]{\arrow};;
      \draw [] (v3\x\z.center) -- (v2\x\w.center) node[sloped, midway, allow upside down]{\arrow};;
    }
    \draw [] (v310.center) -- (v301.center) node[sloped, midway, allow upside down]{\arrow};;
    \draw [] (v341.center) -- (v340.center) node[sloped, midway, allow upside down]{\arrow};;
    \draw [] (v304.center) -- (v314.center) node[sloped, midway, allow upside down]{\arrow};;
    \draw (v241) node[below right=0.4mm] {$H_p^2$};
    \draw (v341) node[right=1.6mm] {$H_p^1$};
    \draw (v230) node[below=0.2mm] {$\mathbf{v}_2$};
    \draw (v330) node[above=1.0mm] {$\mathbf{v}_1$};
    \end{tikzpicture}
    \caption{An $F$-$F$ join.}
  \end{subfigure}
  \begin{subfigure}[b]{.5\textwidth}
    \centering 
    \begin{tikzpicture}
    \usetikzlibrary{positioning}
    \newdimen\nodeDist
    \nodeDist=14mm
    \newcommand{\arrow}{ \tikz \draw[-{Stealth[length=1.5mm, width=1.5mm]}] (-1pt,0) -- (1pt,0); }
    \node at (0,0) (v00) {};
    \foreach \x\y in {0/1, 1/2} {
      \node at ([shift=({60:\nodeDist})]v0\x) (v0\y) {}; 
    }
    \foreach \y in {0, 1, 2} {
      \foreach \x\z in {0/1, 1/2, 2/3, 3/4} {
          \node at ([shift=({0:\nodeDist})]v\x\y) (v\z\y) {};
      }
    }
    \draw [black, fill=black] (v00) circle [radius=0.05] node[below left=1.6mm] {$\mathbf{f}$};
    \draw [black, fill=black] (v10) circle [radius=0.05] node[below=1.6mm] {$\mathbf{a}$} node[below=8mm] {$P_F^1$};
    \draw [black, fill=black] (v20) circle [radius=0.05] node[below right=1.6mm] {$\mathbf{b}$};
    \draw [black, fill=black] (v11) circle [radius=0.05] node[right=1.6mm] {$\mathbf{c}$};
    \draw [black, fill=black] (v02) circle [radius=0.05] node[above=1.4mm] {$\mathbf{d}$};
    \draw [black, fill=black] (v01) circle [radius=0.05] node[left=1.6mm] {$\mathbf{e}$};

    \draw [black, fill=black] (v22) circle [radius=0.05] node[above left=1.6mm] {$\mathbf{f}$};
    \draw [black, fill=black] (v32) circle [radius=0.05] node[above=1.6mm] {$\mathbf{a}$};
    \draw [black, fill=black] (v42) circle [radius=0.05] node[above right=1.6mm] {$\mathbf{b}$};
    \draw [black, fill=black] (v41) circle [radius=0.05] node[right=1.6mm] {$\mathbf{c}$};
    \draw [black, fill=black] (v40) circle [radius=0.05] node[below=1.4mm] {$\mathbf{d}$} node[below=8mm] {$P_T^2$};
    \draw [black, fill=black] (v31) circle [radius=0.05] node[left=1.6mm] {$\mathbf{e}$};

    \draw [] (v00.center) -- (v10.center) node[sloped, midway, allow upside down]{\arrow};;
    \draw [] (v10.center) -- (v20.center) node[sloped, midway, allow upside down]{\arrow};;
    \draw [] (v20.center) -- (v11.center) node[sloped, midway, allow upside down]{\arrow};;
    \draw [] (v11.center) -- (v02.center) node[sloped, midway, allow upside down]{\arrow};;
    \draw [] (v02.center) -- (v01.center) node[sloped, midway, allow upside down]{\arrow};;
    \draw [] (v01.center) -- (v00.center) node[sloped, midway, allow upside down]{\arrow};;
    \draw [] (v10.center) -- (v01.center) node[sloped, midway, allow upside down]{\arrow};;
    \draw [] (v01.center) -- (v11.center) node[sloped, midway, allow upside down]{\arrow};;
    \draw [] (v11.center) -- (v10.center) node[sloped, midway, allow upside down]{\arrow};;
    
    \draw [] (v22.center) -- (v32.center) node[sloped, midway, allow upside down]{\arrow};;
    \draw [] (v32.center) -- (v42.center) node[sloped, midway, allow upside down]{\arrow};;
    \draw [] (v42.center) -- (v41.center) node[sloped, midway, allow upside down]{\arrow};;
    \draw [] (v41.center) -- (v40.center) node[sloped, midway, allow upside down]{\arrow};;
    \draw [] (v40.center) -- (v31.center) node[sloped, midway, allow upside down]{\arrow};;
    \draw [] (v31.center) -- (v22.center) node[sloped, midway, allow upside down]{\arrow};;
    \draw [] (v32.center) -- (v31.center) node[sloped, midway, allow upside down]{\arrow};;
    \draw [] (v31.center) -- (v41.center) node[sloped, midway, allow upside down]{\arrow};;
    \draw [] (v41.center) -- (v32.center) node[sloped, midway, allow upside down]{\arrow};;
    
    \end{tikzpicture}
    \caption{An $F$-$T$ join.}
  \end{subfigure}
\caption{(a) An $F$-$F$ join on $H_p^1$ (solid lines) and $H_p^2$ (dashed lines). Both graphs share the patch $P_F$ (bold lines). (b) The association of vertices between $P_F^1$ and $P_T^2$ in an $F$-$T$ join on $H_p^1$ and $H_p^2$---note the labelling of vertices.} \label{fig:mep_FF_FT}
\end{figure}
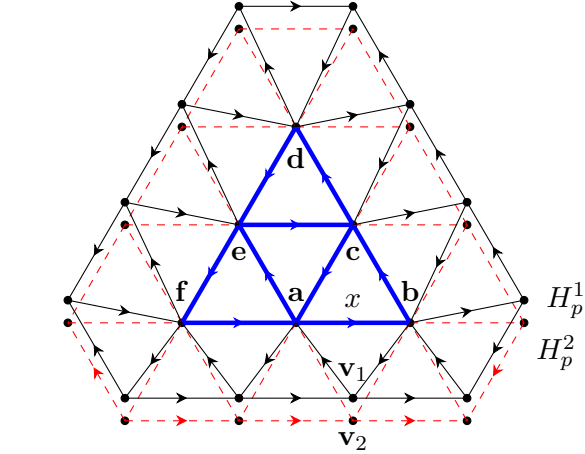
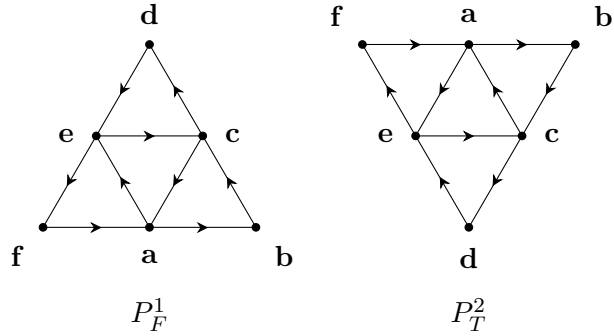

Let $H_p^1$ and $H_p^2$ be two copies of $H_p$, and let $P_F^k$ be an $F$-patch on $H_p^k$ for each $k \in \{1, 2\}$. 
We say that we apply an \emph{$F$-$F$ join} on $(H_p^1, H_p^2)$ if we remove the patches $P_F^1$ and $P_F^2$ on the respective copies and replace them by \emph{one} copy of an $F$-patch $P_F$ (note that this construction is slightly different from the $F$-$F$-$F$ join, as the center of the patch is also included here). 
See \Cref{fig:mep_FF_FT}(a) for an illustration.
We claim that any edge-partition of this new graph into triangles results in \emph{at least one} of $H_p^1$ and $H_p^2$ being partitioned into $T$-triangles. 
Similar to the proof for the $F$-$F$-$F$ join, consider an edge $x = \mathbf{a} \to \mathbf{b}$ belonging to the exterior of $P_F$. 
There are now \emph{three} possible candidates for the third vertex of the triangle containing $x$: $\mathbf{v}_1$, $\mathbf{v}_2$, and $\mathbf{c}$, where $\mathbf{v}_1$ and $\mathbf{v}_2$ are parallel vertices on $H_p^1$ and $H_p^2$, respectively, and $\mathbf{c}$ is a vertex in the center of $P_F$. 
If the third vertex of the triangle is $\mathbf{v}_1$ or $\mathbf{v}_2$, then the same proof for the $F$-$F$-$F$ join can be used to conclude that exactly one of $H_p^1$ and $H_p^2$ is partitioned into $T$-triangles.
Otherwise, the third vertex is $\mathbf{c}$.
In this case, the other edges in the center of the $F$-patch can only belong to $T$-triangles, which implies that both $H_p^1$ and $H_p^2$ can also only be partitioned into $T$-triangles. 
This shows that at least one of $H_p^1$ and $H_p^2$ is partitioned into $T$-triangles. 
It can be verified that these edge-partitions into triangles are indeed valid.

Let $H_p^1$ and $H_p^2$ be two copies of $H_p$, and let $P_F^1$ be an $F$-patch on $H_p^1$ and $P_T^2$ be a $T$-patch on $H_p^2$. 
We say that we apply an \emph{$F$-$T$ join} on $(H_p^1, H_p^2)$ if we remove the patches $P_F^1$ and $P_T^2$ on the respective copies and replace them by one copy of an $F$-patch $P_F$---here, the replacement of $P_T^2$ by $P_F$ is ``mirrored''. 
See \Cref{fig:mep_FF_FT}(b) for an illustration, where the mirroring is across the edge $\mathbf{e} \to \mathbf{c}$. 
We claim that any edge-partition of this new graph into triangles results in $H_p^1$ being partitioned into $T$-triangles or $H_p^2$ being partitioned into $F$-triangles (or both). 
The proof is similar to that of the $F$-$F$ join, except that we reverse the argument regarding $H_p^2$ due to the ``mirror'' effect on $P_T^2$.

We are now ready to prove our result.

\directedtrianglepartitionnphard*

\begin{proof}
We shall reduce \textsc{3SAT} to \textsc{Directed Triangle Partition}. 
Recall that in an instance of \textsc{3SAT}, we are given a set of variables $Y = \{y_1, \ldots, y_q\}$ and a set of clauses $C = \{c_1, \ldots, c_r\}$ where each clause is a disjunction of three literals, i.e., $c_j = \ell_{j, 1} \lor \ell_{j, 2} \lor \ell_{j, 3}$, and each literal is either a variable (i.e., $\ell_{j, k} = y_i$) or its negation (i.e., $\ell_{j, k} = \overline{y_i}$).

Choose $p$ large enough so that there are at least $3r$ $T$-patches and $3r$ $F$-patches in $H_p$ that are pairwise non-interfering (say, $p = 100r$). 
Assign to each variable $y_i$ a separate copy of the graph $Y_i$ isomorphic to $H_p$, and assign to each literal $\ell_{j, k}$ a separate copy of the graph $L_{j, k}$ isomorphic to $H_p$. For each $j$, apply an $F$-$F$-$F$ join on $(L_{j, 1}, L_{j, 2}, L_{j, 3})$ via any $F$-patch in $L_{j, k}$. 
For each $(j, k)$, if the literal $\ell_{j, k}$ corresponds to the variable $y_i$, apply an $F$-$F$ join on $(L_{j, k}, Y_i)$ via any unused $F$-patches, and if the literal $\ell_{j, k}$ corresponds to~$\overline{y_i}$, apply an $F$-$T$ join on $(L_{j, k}, Y_i)$ via any unused $F$-patch and $T$-patch.

Let $G = (V, E)$ denote the constructed graph. This construction can be done in time polynomial in the size of the \textsc{3SAT} instance. 
Note that $G$ is a directed graph with no cycles of length $1$ or $2$. 
We claim that there exists a satisfying assignment in the \textsc{3SAT} instance if and only if $G$ can be edge-partitioned into triangles.

Suppose that there exists a partition of the edges of $G$ into triangles. 
Consider one such partition, and assign $y_i$ as true if and only if $Y_i$ is partitioned into $T$-triangles. 
For each $j$, note that $L_{j, k}$ is partitioned into $F$-triangles for some $k \in \{1, 2, 3\}$ due to the $F$-$F$-$F$ join---we claim that the corresponding literal $\ell_{j, k}$ is satisfied.
If $\ell_{j, k} = y_i$ for some $i$, then $Y_i$ must be partitioned into $T$-triangles by the $F$-$F$ join on $(L_{j, k}, Y_i)$, which means that $\ell_{j, k} = y_i$ is true. 
If $\ell_{j, k} = \overline{y_i}$ for some $i$, then $Y_i$ must be partitioned into $F$-triangles by the $F$-$T$ join on $(L_{j, k}, Y_i)$, which means that $y_i$ is false and $\ell_{j, k} = \overline{y_i}$ is true. 
In both cases, we see that the literal $\ell_{j, k}$ is satisfied.

Conversely, suppose there exists a satisfying assignment in the \textsc{3SAT} instance, and consider any satisfying assignment.
For each $i$, if $y_i$ is true, partition $Y_i$ into $T$-triangles; else, partition $Y_i$ into $F$-triangles. 
For each $j$, at least one of the literals in $c_j$ is true; pick any one of them, say $\ell_{j, k}$, and partition $L_{j, k}$ into $F$-triangles, and partition the other two $L_{j, k'}$ into $T$-triangles. 
We now verify that the edge-partition is a valid partition by checking that the restrictions caused by the joins are not violated.
For each $j$, consider the $F$-$F$-$F$ join on $(L_{j, 1}, L_{j, 2}, L_{j, 3})$---since one $L_{j, k}$ is edge-partitioned into $F$-triangles and the other two into $T$-triangles, the requirement on the $F$-$F$-$F$ join is satisfied.
Now, for each $(j, k)$, if the literal $\ell_{j, k}$ corresponds to the variable $y_i$, then the join on $(L_{j, k}, Y_i)$ is $F$-$F$, and at least one of $L_{j, k}$ and $Y_i$ is partitioned into $T$-triangles (otherwise, if both are partitioned into $F$-triangles, then $y_i$ is false and $\ell_{j, k}$ is true, which is not possible).
On the other hand, if the literal $\ell_{j, k}$ corresponds to $\overline{y_i}$, then the join on $(L_{j, k}, Y_i)$ is $F$-$T$, and $L_{j, k}$ is partitioned into $T$-triangles or $Y_i$ is partitioned into $F$-triangles, or both (otherwise, if $L_{j, k}$ is partitioned into $F$-triangles and $Y_i$ is partitioned into $T$-triangles, then both $\overline{y_i} = \ell_{j, k}$ and $y_i$ are true, which is not possible). 
Therefore, the edges of $G$ can be partitioned into triangles.
\end{proof}

\end{document}